\documentclass[aps,pra,notitlepage,superscriptaddress,10pt]{revtex4-2}

\usepackage{aliascnt,amsmath,amssymb,amsthm,booktabs,colortbl,diagbox,dsfont,hyperref,lmodern,mathtools,multirow,xcolor}
\usepackage[UKenglish]{babel}
\usepackage[capitalize,nameinlink]{cleveref}
\usepackage[T1]{fontenc}
\usepackage[scaled=0.8]{FiraMono}
\usepackage[caption=false]{subfig}

\definecolor{darkred}{rgb}{0.5,0,0}
\definecolor{darkgreen}{rgb}{0,0.5,0}
\definecolor{darkblue}{rgb}{0,0,0.5}
\hypersetup{colorlinks,breaklinks,linkcolor=darkred,citecolor=darkgreen,urlcolor=darkblue}

\DeclareMathOperator{\Tr}{Tr}

\newtheorem{theorem}{Theorem}
\newaliascnt{proposition}{theorem}
\newtheorem{proposition}[proposition]{Proposition}
\aliascntresetthe{proposition}
\newaliascnt{lemma}{theorem}
\newtheorem{lemma}[lemma]{Lemma}
\aliascntresetthe{lemma}
\newaliascnt{corollary}{theorem}

\aliascntresetthe{corollary}
\theoremstyle{definition}
\newaliascnt{definition}{theorem}
\newtheorem{definition}[definition]{Definition}
\aliascntresetthe{definition}
\newaliascnt{conjecture}{theorem}
\newtheorem{conjecture}[conjecture]{Conjecture}
\aliascntresetthe{conjecture}
\Crefname{proposition}{Proposition}{Propositions}
\Crefname{lemma}{Lemma}{Lemmata}
\Crefname{corollary}{Corollary}{Corollaries}
\Crefname{definition}{Definition}{Definitions}
\Crefname{conjecture}{Conjecture}{Conjectures}
\Crefname{appendix}{Appendix}{Appendices}

\renewcommand{\geq}{\geqslant}
\renewcommand{\leq}{\leqslant}

\newcommand{\algrel}[2]{\mathcal{R}_{#1,#2}}
\newcommand{\abs}[1]{\lvert#1\rvert}
\newcommand{\bC}{\mathbb C}
\newcommand{\cc}{\cellcolor[gray]{0.9}}
\newcommand{\cd}{\cellcolor[gray]{0.8}}
\newcommand{\ce}{\cellcolor[gray]{0.7}}
\newcommand{\cf}{\cellcolor[gray]{0.6}}
\newcommand{\cH}{\mathcal H}
\newcommand{\cK}{\mathcal K}
\newcommand{\fS}{\mathfrak S}
\newcommand{\id}{\mathds{1}}
\newcommand{\mg}{\mathrm{g}}
\newcommand{\stackcell}[2]{\genfrac{}{}{0pt}{}{#1}{#2}}
\newcommand{\vj}{\vec{\jmath}}

\makeatletter
\def\CT@@do@color{%
 \global\let\CT@do@color\relax
 \@tempdima\wd\z@
 \advance\@tempdima\@tempdimb
 \advance\@tempdima\@tempdimc
 \advance\@tempdimb0.9\tabcolsep
 \advance\@tempdimc\tabcolsep
 \advance\@tempdima2\tabcolsep
 \kern-\@tempdimb
 \leaders\vrule
 \hskip\@tempdima\@plus 1fill
 \kern-\@tempdimc
\hskip-\wd\z@ \@plus -1fill }
\makeatother

\begin{document}

\author{Sébastien Designolle}
\email{sebastien.designolle@inria.fr}
\affiliation{Inria, ENS de Lyon, UCBL, LIP, 69342, Lyon Cedex 07, France}

\author{Máté Farkas}
\email{mate.farkas@york.ac.uk}
\affiliation{Department of Mathematics, University of York, York, YO10 5GH, United Kingdom}

\title{\texorpdfstring{$k$}{k}-fold unbiased measurements and maximal incompatibility}
\date{17th September 2026}

\begin{abstract}
 Mutually unbiased bases capture perfect complementarity between two quantum measurements.
 Extensions beyond pairwise unbiasedness have been proposed, but essentially no non-trivial higher-order constructions are known.
 We introduce $k$-fold unbiased measurements ($k$-UMs), extending the $k$-fold unbiased bases notion of~[\href{https://arxiv.org/abs/1706.04446}{\texttt{arXiv:1706.04446}}] from rank-one basis measurements to arbitrary-rank projective measurements, and show that this higher-rank setting supports a much richer theory.
 We develop the notions of algebraic and spectral $k$-UMs and prove that they coincide for rank-one measurements and triples of measurements (3-UMs).
 We establish strong no-go results for higher-order rank-one constructions and three-outcome 3-UMs, but obtain infinitely many higher-rank triples using Hadamard matrices and Clifford algebras.
 We then give these structures an exact operational interpretation in terms of measurement incompatibility.
 For 3-UMs with any number of outcomes, we determine their generalised incompatibility robustness exactly and construct an explicit joint measurement for their noisy versions at the compatibility threshold.
 Finally, we implement the symmetry reduction of the sum-of-squares hierarchy recently introduced in~[\href{https://doi.org/10.1088/1367-2630/ae72ba}{\textit{New J.~Phys.}~{\bf 28}, 064509 (2026)}] and give numerical evidence that the noise thresholds arising from the $k$-UM analysis may characterise the asymptotic behaviour of this hierarchy.
 In particular, with very high precision, we numerically show that our constructed four-outcome 3-UMs are among the most incompatible triples of four-outcome measurements.
\end{abstract}

\maketitle

\section{Introduction}

The impossibility of jointly realising arbitrary measurements is one of the basic departures of quantum theory from classical physics.
This \emph{incompatibility} of measurements underlies fundamental quantum phenomena such as Bell non-locality and Einstein--Podolsky--Rosen steering, and every incompatible tuple of measurements offers an advantage in a suitable state-discrimination task~\cite{HMZ16,GHK+23,BCP+14,UCNG20,QVB14,UMG14,UBGP15,SSC19,CHT19,UKS+19,OB19}.
The quantitative study of measurement incompatibility is therefore both foundational and operational.
In particular, the equivalence between joint measurability and unsteerability turns noise thresholds for measurements into one-sided device-independent thresholds for steering experiments~\cite{QVB14,UMG14,UBGP15,CS16,DSR+21}.

Mutually unbiased bases (MUBs) provide one of the canonical examples of incompatible measurements~\cite{Sch60,Iva81,WF89}.
Two orthonormal bases $\{|e_a\rangle\}_{a=1}^n$ and $\{|f_b\rangle\}_{b=1}^n$ on $\mathbb{C}^n$ are mutually unbiased when $\abs{\langle e_a|f_b\rangle}^2=1/n$ for every $a,b$.
MUBs have long played a central role in tomography, uncertainty relations, cryptography, and finite-dimensional phase-space constructions~\cite{WF89,DEBZ10,MW26}.
Their associated rank-one projective measurements are also extremal for several incompatibility questions.
For pairs, the most incompatible measurements are now well understood both when the Hilbert-space dimension is fixed~\cite{BHSS13,DFK19} and when the number of outcomes is fixed~\cite{TFR+21,Des26}.
The latter setting naturally leads from MUBs to mutually unbiased measurements (MUMs), whose measurement operators satisfy the same algebraic relations as those of MUBs while allowing arbitrary rank~\cite{TFR+21,FKN23}.
Unlike MUBs, arbitrarily many pairwise MUMs exist for a fixed number of outcomes~\cite{FKN23}.

The situation changes sharply when unbiasedness conditions involve more than two measurements.
One of us introduced $k$-fold unbiased bases ($k$-UBs) as a higher-order extension of the MUB condition motivated by quantum random access codes~\cite{Far17}.
The definition fixes sums of cyclic Bargmann invariants of selected basis vectors and is hereditary, so every sub-$(k-1)$-tuple of a $k$-UB is again a $(k-1)$-UB for $k\geq3$.
This rigidity is attractive, but it also makes existence of $k$-UBs exceptional, the only known example with $k\geq3$ being a 3-UB in dimension two, corresponding to three MUBs.
To overcome this difficulty in constructing $k$-UBs, we draw on the operational success of MUMs, replacing vector overlaps by algebraic relations on the measurement operators and allowing higher rank.
The first purpose of this work is to develop this extension carefully and to separate the algebraic information inherited from the $k$-UB definition from the spectral information that turns out to be relevant to incompatibility bounds.

The second purpose is to connect $k$-UMs to maximal measurement incompatibility.
For a fixed tuple of measurements $\mathsf A$, the generalised incompatibility robustness $\eta^{\mg}(\mathsf A)$ quantifies the largest fraction of the measurements that can remain after adding arbitrary noise while obtaining a compatible tuple~\cite{DFK19}.
This means that the lower $\eta^{\mg}(\mathsf A)$ is, the more incompatible is the tuple $\mathsf A$.
The universal optimisation $\chi_k^{\mg}(n)$ is then defined as the infimum of $\eta^{\mg}$ over all $k$-tuples of $n$-outcome measurements in arbitrary finite dimension.

Our results can be summarised as follows.
First, we introduce algebraic and spectral $k$-fold unbiased measurements ($k$-UMs).
The definition of algebraic $k$-UMs is the operator-valued lift of the $k$-UB condition~\cite{Far17}; it is hereditary down to pairs and thus implies pairwise MUM relations.
Spectral $k$-UMs are defined in terms of the universal polynomial $\mu_{k,n}(t)=(1-n^{-1}\mathrm d/\mathrm dt)^kt^n$.
For rank-$r$ $n$-outcome projective measurements, spectral $k$-UMs are precisely those for which the characteristic polynomial of every sum of $k$ projections selected from $k$ different measurements is $\mu_{k,n}(t)^r$.
We show that at rank one, algebraic and spectral $k$-UMs agree.
At arbitrary rank, algebraic and spectral $k$-UMs agree for pairs, triples, and dichotomic $k$-tuples.
No general equivalence is proved for $k\geq4$ and $n>2$.

Second, we obtain $k$-UM non-existence and construction results.
There are no rank-one algebraic or spectral $k$-UMs for $k\geq3$, except for $(k,n)=(3,2)$.
Moreover, no finite-dimensional 3-UM with three outcomes exists at any rank.
On the positive side, a real Hadamard matrix and a complex Clifford module produce a 3-UM.
The smallest non-binary instance has four outcomes, rank two, and dimension eight; it is the complexification of a quaternionic three-fold unbiased triple.

Third, we determine the incompatibility of every 3-UM.
In the dual formulation of the optimisation that computes the incompatibility robustness~\cite{DFK19}, writing $\lambda_{k,n}$ for the largest zero of $\mu_{k,n}$ gives $\eta^{\mg}\leq\lambda_{k,n}/k$ for all spectral $k$-UMs.
In the primal formulation, for $k=3$ the joint measurement using the construction from~\cite{DSFB19} reaches the same value, proving its tightness.
Consequently every 3-UM with $n$ outcomes satisfies $\eta^{\mg}=\lambda_{3,n}/3$.
For $n=4$, this gives $[1+\cos(\pi/9)]/3$ for the explicit rank-two construction.
For higher $k$, we identify conditions for the primal construction to attain the dual upper bound, but a general proof for every algebraic $k$-UM is not presently available.

Finally, we study maximal incompatibility of $k$-UMs.
The defining properties of spectral $k$-UMs make them good candidates for minimising the upper bound on the generalised robustness presented in~\cite{DFK19}.
We show this optimality for rank-one measurements and provide some generalisations for arbitrary rank in specific cases.
Unfortunately, minimising this upper bound does not by itself solve the primal question of maximal incompatibility.
To address this problem, we use the sum-of-squares hierarchy from~\cite{Des26}, strengthened by symmetry reduction~\cite{GP04,IR22}.
With high precision, the numerical universal lower bound for $(k,n)=(3,4)$ agrees with the explicit value, providing strong numerical evidence that the Hadamard--Clifford construction is maximally incompatible in this case.
We conjecture more generally that the hierarchy converges to $\lambda_{k,n}/k$ for every $k$ and $n$.

The article is organised as follows.
\cref{sec:background} recalls MUMs and incompatibility robustness.
\cref{sec:kums} introduces the notions of algebraic and spectral $k$-UMs, proves the main structural relations, states the no-go results, and presents the Hadamard--Clifford construction.
\cref{sec:fixed} determines the robustness of fixed $k$-UMs in the proved regimes, with the exact triple result as the central consequence.
\cref{sec:universal} treats global optimality of $k$-UMs and the numerical universal hierarchy.
Proofs and computational details are collected in the appendices.

\section{Background}
\label{sec:background}

\subsection{Mutually unbiased measurements}

An $n$-outcome positive operator-valued measure (POVM), or simply measurement, on a finite-dimensional Hilbert space $\cH$ is a tuple $A=(A_a)_{a=1}^n$ of positive semidefinite operators summing to the identity.
It is projective, or a projection-valued measure (PVM), when every $A_a$ is a projection.
Two measurements $A=(A_a)_a$ and $B=(B_b)_b$ are mutually unbiased measurements (MUM) when
\begin{equation}
 A_aB_bA_a=\frac1nA_a
 \qquad\text{and}\qquad
 B_bA_aB_b=\frac1nB_b
 \label{eq:mum}
\end{equation}
for all outcomes $a,b$~\cite{TFR+21,FKN23}.
An equivalent complementarity formulation appeared earlier in~\cite{TSWR18}.
We follow here the terminology introduced in~\cite{TFR+21} and developed in~\cite{FKN23}; it should not be confused with the earlier, inequivalent notion of MUMs of Kalev and Gour, defined through Hilbert--Schmidt overlaps with an efficiency parameter and not restricted to projective measurements~\cite{KG14}.
Note that MUMs must be projective: summing the first identity over $b$ gives $A_a^2=A_a$, and summing the second over $a$ gives $B_b^2=B_b$.
Taking traces in both identities then shows that every projection has the same non-zero rank.
When this rank is one, \cref{eq:mum} is exactly the MUB condition.

MUMs retain the operational complementarity of MUBs: measuring $(B_b)_b$ on an eigenstate of $A_a$ yields a uniformly random outcome~\cite{TFR+21}.
But MUMs are genuinely different from MUBs: there exist pairs of MUMs that cannot be mapped to a pair of MUBs by any completely positive unital map~\cite{TFR+21}.
Quaternionic Hadamard matrices can give rise to such generic constructions, and the number of pairwise MUMs with a fixed outcome number is unbounded~\cite{FKN23}.
This flexibility in MUM constructions is the main reason to seek $k$-fold unbiasedness at arbitrary rank.

\subsection{Generalised incompatibility robustness}

Writing $[m]\coloneqq\{1,\ldots,m\}$, a tuple of POVMs $\mathsf A=(A_{a|x})_{a\in[n],x\in[k]}$ is compatible if there is a ``parent'' POVM $(G_{\vj})_{\vj\in[n]^k}$ whose marginals reproduce every measurement, $\sum_{\vj}\delta_{j_x,a}G_{\vj}=A_{a|x}$ where $\delta$ is the Kronecker delta~\cite{HMZ16,GHK+23}.
Otherwise, the tuple is called incompatible.
The degree of incompatibility of a tuple can be quantified by the amount of ``noise'' needed to be added to the measurements before they become compatible.
Two standard noise models are depolarising robustness, in which every outcome is mixed with the effect $(\Tr{A_{a|x}})\id/d$, and generalised robustness, in which arbitrary measurement noise is allowed; we focus on the latter~\cite{HKS15,DFK19}.

The generalised incompatibility robustness is defined as the following primal semidefinite programming (SDP) instance:
\begin{equation}
 \eta^{\mg}(\mathsf A)\coloneqq
 \max\left\{\eta:\ \sum_{\vj}G_{\vj}=\id,\ G_{\vj}\succeq0,\
 \sum_{\vj}\delta_{j_x,a}G_{\vj}\succeq\eta A_{a|x}\ \text{for all }a,x\right\}.
 \label{eq:eta}
\end{equation}
Here $X\succeq Y$ means that $X-Y$ is positive semidefinite.
The inequality form is equivalent to mixing each $A_{a|x}$ with an arbitrary measurement so that the resulting tuple is compatible~\cite{DFK19}.
Smaller values correspond to greater incompatibility.
We refer to the optimisation variable $\eta$ as the visibility.

For fixed $k$ and $n$, we define the universal value
\begin{equation}
 \chi_k^{\mg}(n)\coloneqq\inf_{d\geq1}\ \min_{\mathsf A\in\mathrm{POVM}_n(\mathbb C^d)^k}\eta^{\mg}(\mathsf A),
 \label{eq:chi}
\end{equation}
where $\mathrm{POVM}_n(\mathbb C^d)$ denotes the set of $n$-outcome POVMs on $\mathbb C^d$.
By pre-processing monotonicity of the generalised robustness together with Naimark dilation, the infimum in \cref{eq:chi} may be restricted to projective measurements without changing its value~\cite{DFK19,Des26}.

Measurement incompatibility and steering are equivalent at the qualitative level~\cite{QVB14,UMG14,UBGP15} and admit quantitative robustness correspondences~\cite{CS16,DFK19}.
Consequently, an exact value of $\chi_k^{\mg}(n)$ provides a dimension-independent benchmark for multisetting steering, while a fixed-tuple value gives the noise threshold of a concrete implementation.

\section{\texorpdfstring{$k$}{k}-fold unbiased measurements}
\label{sec:kums}

\subsection{Definitions and structural relations}

We begin with the notion of algebraic $k$-UMs.
Let $\mathsf A=(A_{a|x})_{a\in[n],x\in[k]}$ be a $k$-tuple of $n$-outcome measurements.

\begin{definition}[Algebraic $k$-UM]
 For $k,n\geq2$, a tuple $\mathsf A$ is an \emph{algebraic $k$-fold unbiased measurement ($k$-UM)} if, for every fixed setting $x$ and every outcome string $\vj=(j_1,\ldots,j_k)$,
 \begin{equation}
  \sum_{(x_2,\ldots,x_k)}
  A_{j_x|x}A_{j_{x_2}|x_2}\cdots A_{j_{x_k}|x_k}A_{j_x|x}
  =\frac{(k-1)!}{n^{k-1}}A_{j_x|x},
  \label{eq:algebraic}
 \end{equation}
 where the sum runs over all $(k-1)!$ orderings of $[k]\setminus\{x\}$.
\end{definition}

The motivation behind this definition is to extend the $k$-UB definition from~\cite{Far17} to arbitrary rank.
Indeed, for rank-one projective measurements the two definitions are equivalent.
Later, \cref{thm:relations} shows that this notion is hereditary, just like the $k$-UB definition.
In particular, $k$-UMs are pairwise MUMs and therefore projective and of equal rank.

We now introduce the notion of spectral $k$-UMs in a way that first seems unrelated to algebraic $k$-UMs.
Let $\mathsf P=(P_{a|x})_{a\in[n],x\in[k]}$ be a $k$-tuple of $n$-outcome projective measurements on $\cH$ with common rank $r$, so $\dim\cH=nr$.
For $\vj=(j_1,\ldots,j_k)$, define the \emph{selected sum} $S_{\vj}\coloneqq\sum_{x=1}^kP_{j_x|x}$.
For any such tuple of rank-one projective measurements ($r=1$), we show in \cref{eq:average-char} that averaging the characteristic polynomials of all $n^k$ selected sums yields the polynomial
\begin{equation}
 \mu_{k,n}(t)
 \coloneqq\sum_{\ell=0}^{\min\{k,n\}}\frac{(-1)^{\ell}\ell!}{n^\ell}\binom{k}{\ell}\binom{n}{\ell}t^{n-\ell}
 =\left(1-\frac1n\frac{\mathrm d}{\mathrm dt}\right)^kt^n.
 \label{eq:mu}
\end{equation}
For pairs and triples of measurements, we have the explicit expressions
\begin{equation}
 \mu_{2,n}(t)=t^{n-2}\left(t^2-2t+\frac{n-1}{n}\right)
 \qquad\text{and}\qquad
 \mu_{3,n}(t)=t^{n-3}\left(t^3-3t^2+\frac{3(n-1)}nt-\frac{(n-1)(n-2)}{n^2}\right).
 \label{eq:mu3}
\end{equation}
For a pair of MUMs with rank $r$, every selected sum has non-zero eigenvalues $1\pm1/\sqrt{n}$ and characteristic polynomial $\mu_{2,n}^r$; see \cref{eq:mub-pair-char}.
This motivates the following definition, which fixes the spectrum of each selected sum.

\begin{definition}[Spectral $k$-UM]

 For $k,n\geq2$, a projective tuple $\mathsf P$ with common outcome rank $r$ is a \emph{spectral $k$-UM} if, for every outcome string $\vj$,
 \begin{equation}
  \det(t\id-S_{\vj})=\mu_{k,n}(t)^r.
  \label{eq:spectral}
 \end{equation}
\end{definition}

Denoting by $\mu'_{\ell,n}$ the derivative of $\mu_{\ell,n}$, the recursion $\mu_{\ell+1,n}=\mu_{\ell,n}-\mu'_{\ell,n}/n$ gives a proof that all zeros of $\mu_{k,n}$ are real and non-negative, with a zero of multiplicity $n-k$ at the origin when $k<n$ and all positive zeros simple; the details are given in \cref{app:spectral}.
We write $\lambda_{k,n}$ for the largest zero.
The coefficient expansion in \cref{eq:mu} also uncovers an elegant symmetry between $k$ and $n$: for $n\geq k$, direct comparison of coefficients gives $n^k\mu_{k,n}(t)=k^kt^{n-k}\mu_{n,k}\left(nt/k\right)$.
The positive roots are therefore mapped by $t\mapsto nt/k$, and in particular
\begin{equation}
 \frac{\lambda_{k,n}}k=\frac{\lambda_{n,k}}n.
\end{equation}

\begin{theorem}
 \label{thm:relations}
 The following statements hold.
 \begin{enumerate}
  \item\label{itm:relations-heredity} Every algebraic $k$-UM is a projective tuple with equal outcome rank and is hereditary: every sub-$\ell$-tuple ($\ell\geq2$) is an algebraic $\ell$-UM.
  In particular, every algebraic $k$-UM is pairwise MUM.
  \item\label{itm:relations-rank-one} At rank one, algebraic and spectral $k$-UMs coincide.
  \item\label{itm:relations-pairs} For $k=2$ at arbitrary rank, algebraic and spectral 2-UMs coincide with pairs of MUMs.
  \item\label{itm:relations-binary} For $n=2$ at arbitrary rank, algebraic and spectral $k$-UMs coincide with the projective tuples generated by pairwise anticommuting Hermitian unitary observables.
  \item\label{itm:relations-triples} For $k=3$ at arbitrary rank, algebraic and spectral 3-UMs coincide.
 \end{enumerate}
\end{theorem}

The proof is given in \cref{app:relations}.
In the following, whenever algebraic and spectral $k$-UMs coincide, we simply refer to them as $k$-UMs.

For later use, we explicitly write the defining relation of algebraic 3-UMs: for three measurements $P=(P_a)_a$, $Q=(Q_b)_b$, and $R=(R_c)_c$, \cref{eq:algebraic} is equivalent to
\begin{equation}
 P_a(Q_bR_c+R_cQ_b)P_a=\frac{2}{n^2}P_a,
 \qquad
 Q_b(P_aR_c+R_cP_a)Q_b=\frac{2}{n^2}Q_b,
 \qquad\text{and}\qquad
 R_c(P_aQ_b+Q_bP_a)R_c=\frac{2}{n^2}R_c
 \label{eq:algebraic-three}
\end{equation}
for all outcomes $a,b,c$.
Together with pairwise mutual unbiasedness, these identities fix the relative overlaps of the three selected outcome subspaces sufficiently to determine the spectrum of every sum $P_a+Q_b+R_c$; see \cref{app:relations}.

Although we restrict attention to the case where the number of measurements coincides with the order of unbiasedness, the notation extends naturally.
For $2\leq\ell\leq k$, we call a $k$-tuple of measurements a $(k,\ell)$-UM if every sub-$\ell$-tuple is an $\ell$-UM, specifying algebraic or spectral when relevant.
For instance, the projective measurements associated with a complete set of MUBs in dimension $d$ form a $(d+1,2)$-UM~\cite{WF89}.
When $k=\ell$, we use the shorthand $k$-UM throughout this article.

\subsection{No-go theorems}

As already noted in~\cite{Far17}, the $k$-UB notion (rank-one $k$-UM) is very rigid, and examples are difficult to come by.
Here, we fully resolve the rank-one case, confirming that the only instances are those already known.

\begin{theorem}[Rank-one no-go]
 \label{thm:rank-one-no-go}
 For $k,n\geq2$, a rank-one (algebraic or spectral) $k$-UM exists only when $k=2$ with arbitrary $n$, or when $(k,n)=(3,2)$.
\end{theorem}

The proof is given in \cref{app:nogo}.
This theorem shows why passing from bases to higher-rank measurements is not simply a technical generalisation but a necessity for non-binary triples.

The first candidate beyond the rank-one regime is therefore a triple of three-outcome measurements.
Unfortunately, the 3-UM condition remains impossible to realise even at arbitrary rank.

\begin{theorem}[Three-outcome no-go]
 \label{thm:qutrit-no-go}
 There is no algebraic or spectral 3-UM with three outcomes on a non-zero finite-dimensional Hilbert space.
 By heredity, there is no algebraic $k$-UM with three outcomes for any $k\geq3$.
\end{theorem}

\cref{app:nogo} gives the proof of the obstruction for algebraic $k$-UMs; the equivalence in \cref{thm:relations} then rules out spectral 3-UMs as well.
Together, \cref{thm:rank-one-no-go,thm:qutrit-no-go} identify $(k,n,r)=(3,4,2)$ as the first viable non-binary target.
In \cref{subsec:Hadamard-Clifford} below, we give a realisation of this case.

To finish with no-go results, we present a packing bound which applies to both definitions.
\begin{proposition}
 \label{prop:packing}
 Every algebraic or spectral $k$-UM with $n$ outcomes of rank $r$ satisfies
 \begin{equation}
  k(n-1)\leq(nr)^2-1.
  \label{eq:packing-main}
 \end{equation}
\end{proposition}
The complete proof is given in \cref{app:nogo} and relies on the following argument.
The defining property of spectral $k$-UMs forces pairwise trace-unbiasedness, that is, $\Tr(P_{a|x}P_{b|y})=r/n$ for $x\neq y$.
It follows that the $(n-1)$-dimensional subspaces spanned by $\{P_{a|x}-\id/n:a\in[n]\}$ are mutually orthogonal for distinct $x$.
The standard dimension-counting argument for MUBs then gives the claimed bound~\cite{WF89,DEBZ10,MW26}.
No non-binary higher-order construction saturating \cref{eq:packing-main} is presently known.

\subsection{Hadamard--Clifford construction of 3-UMs}
\label{subsec:Hadamard-Clifford}

The first viable non-binary case identified above, $(k,n,r)=(3,4,2)$, can be realised using a construction based on real Hadamard matrices and complex Clifford algebras~\cite{Hor07,LM89}.
More generally, the same construction produces 3-UMs for every outcome number $n$ for which a real Hadamard matrix exists.

A real Hadamard matrix of order $n$ is a matrix $H=(h_{ab})_{a,b}\in\{\pm1\}^{n\times n}$ satisfying $HH^{\mathsf T}=n\id_n$.
Given such a matrix and a tuple $E=(E_a)_{a\in[n]}$ of unitaries on $\bC^r$, define the operators $V_b(H,E):\bC^r\to\bC^n\otimes\bC^r$ given by, for $b\in[n]$,
\begin{equation}
 V_b(H,E)\lvert\psi\rangle
 \coloneqq\frac1{\sqrt n}\sum_{a\in[n]}h_{ab}\lvert a\rangle\otimes E_a\lvert\psi\rangle.
\end{equation}
The corresponding projections define a measurement $M(H,E)$ with elements
\begin{equation}
 M_b(H,E)\coloneqq V_b(H,E)V_b(H,E)^*.
 \label{eq:hadamard-measurement}
\end{equation}
Indeed, the Hadamard identity gives $V_b(H,E)^*V_{b'}(H,E)=\delta_{bb'}\id$, so $M(H,E)$ is an $n$-outcome rank-$r$ PVM.
When all $E_a=\id$, we simply write $V_b(H)$ and $M(H)$.

\begin{theorem}[Hadamard--Clifford construction]
 \label{thm:construction}
 Let $H=(h_{ab})$ and $\widetilde H=(\widetilde h_{ac})$ be real Hadamard matrices of order $n$.
 Let $E=(E_a)_{a\in[n]}$ be a tuple of unitaries on $\bC^r$ satisfying
 \begin{equation}
  E_a^*E_{a'}+E_{a'}^*E_a=2\delta_{aa'}\id
  \qquad\text{and}\qquad
  E_aE_{a'}^*+E_{a'}E_a^*=2\delta_{aa'}\id.
  \label{eq:clifford-unitaries}
 \end{equation}
 On $\bC^n\otimes\bC^r$, define $P=(\lvert a\rangle\langle a\rvert\otimes\id)_{a\in[n]}$, $Q=M(H)$, and $R=M(\widetilde H,E)$.
 Then $(P,Q,R)$ is a (spectral and algebraic) 3-UM with $n$ outcomes of rank $r$.
\end{theorem}

The calculation is deferred to \cref{app:construction}.
Taking $E_1=\id$ and $E_a=\mathrm i\Gamma_{a-1}$ for $a=2,\ldots,n$, where the $\Gamma_j$ form a Hermitian representation of the complex Clifford algebra on $n-1$ generators, verifies \cref{eq:clifford-unitaries}~\cite{LM89}.
An irreducible complex Clifford representation on $n-1$ generators has dimension $2^{\lfloor(n-1)/2\rfloor}$, and every finite-dimensional representation is a direct sum of irreducible ones.
Thus, whenever a real Hadamard matrix of order $n$ exists, the construction works for every rank divisible by $2^{\lfloor(n-1)/2\rfloor}$.
This includes every order $n=2^s$, as $\bigl(\begin{smallmatrix}1&1\\1&-1\end{smallmatrix}\bigr)^{\otimes s}$ is a real Hadamard matrix, and therefore works for infinitely many outcome numbers, albeit with an auxiliary rank growing exponentially in $n$.
More generally, numerous constructions are known, and the published literature establishes real Hadamard matrices for every admissible (multiple-of-four) order $n<668$, leaving only $n=668,716,892,1132$ unresolved up to $1208$~\cite{CP25}.

For $n=4$, take $H=\widetilde H=H_4$, where
\begin{equation}
 H_4\coloneqq\begin{pmatrix}1&1\\1&-1\end{pmatrix}^{\otimes2}=\begin{pmatrix}1&1&1&1\\1&-1&1&-1\\1&1&-1&-1\\1&-1&-1&1\end{pmatrix}
 \quad\text{and choose}\quad
 (E_1,E_2,E_3,E_4)=(\id_2,\mathrm i\sigma_x,-\mathrm i\sigma_y,\mathrm i\sigma_z).
 \label{eq:four-outcome-3um}
\end{equation}
\cref{thm:construction} then defines three four-outcome rank-two PVMs on $\mathbb C^8$ forming a 3-UM.
These measurements are the complexification of three quaternionic mutually unbiased bases and satisfy \cref{eq:mum,eq:algebraic-three}~\cite{FKN23}.

The same quaternionic construction can be extended to four measurements such that every sub-3-tuple is a 3-UM, giving a $(4,3)$-UM with four outcomes of rank two in dimension eight; see \cref{app:construction}.

\section{Incompatibility robustness of fixed \texorpdfstring{$k$}{k}-UMs}
\label{sec:fixed}

\subsection{Selected-sum upper bound}

First introduced in~\cite{DSFB19} for the depolarising incompatibility robustness, the selected-sum upper bound was explicitly stated in~\cite[Appendix~E]{DFK19} for the generalised robustness and reads, for a $k$-tuple $\mathsf P$ of projective measurements,
\begin{equation}
 \eta^{\mg}(\mathsf P)\leq\frac1k\max_{\vj}\left\|S_{\vj}\right\|_\infty,
 \label{eq:selected-upper}
\end{equation}
where $\|\cdot\|_\infty$ denotes the operator norm.
A spectral $k$-UM $\mathsf P$ with $n$ outcomes satisfies $\|S_{\vj}\|_\infty=\lambda_{k,n}$ for every $\vj\in[n]^k$, and therefore
\begin{equation}
 \eta^{\mg}(\mathsf P)\leq\frac{\lambda_{k,n}}k.
\end{equation}
The next subsection studies whether a parent measurement reaches this upper bound.

\subsection{Maximal-eigenspace parent}

Let $\Pi_{\vj}$ be the spectral projection of $S_{\vj}$ at eigenvalue $\lambda_{k,n}$, taken to be zero when that eigenvalue is absent.
The use of these maximal eigenspaces extends the parent constructions for pairs of MUBs~\cite{ULMH16} and for larger MUB sets and other symmetric measurements~\cite{DSFB19,NDBG20}.

\begin{theorem}
 \label{thm:maximal-eigenspace-parent}
 Let $\mathsf P$ be a $k$-tuple of $n$-outcome projective measurements such that $\max_{\vj}\|S_{\vj}\|_\infty=\lambda_{k,n}$.
 If
 \begin{equation}
  \frac1{n^{k-1}}\sum_{\vj}\delta_{j_x,a}\Pi_{\vj}
  =\frac{\lambda_{k,n}}kP_{a|x}+\left(1-\frac{\lambda_{k,n}}{k}\right)\frac{\id-P_{a|x}}{n-1}
  \label{eq:conditional-projector}
 \end{equation}
 for every setting $x$ and outcome $a$, then $\eta^{\mg}(\mathsf P)=\lambda_{k,n}/k$, with optimal parent
 \begin{equation}
  G_{\vj}=\frac{\Pi_{\vj}}{n^{k-1}}.
  \label{eq:maximal-eigenspace-parent}
 \end{equation}
\end{theorem}

The proof is given in \cref{app:primal}.
Spectral $k$-UMs satisfy the first hypothesis by definition, so, in order to apply \cref{thm:maximal-eigenspace-parent} to them, it remains to check in which cases \cref{eq:conditional-projector} holds.
\cref{app:primal} further explains how the maximal-eigenvalue projector is a low-degree Lagrange polynomial in $S_{\vj}$, computes the averages $n^{1-k}\sum_{\vj}\delta_{j_x,a}S_{\vj}^m$ for low $m$, and identifies regimes where \cref{eq:conditional-projector} holds.

Note that \cref{eq:conditional-projector} is stronger than the marginal constraints in the generalised-robustness SDP \cref{eq:eta}: it fixes each marginal exactly rather than requiring an operator inequality.
Interpreting the right-hand side as random or depolarising noise, \cite[Appendix~E]{DFK19} therefore implies that the random and depolarising incompatibility robustnesses coincide, with common value $(\lambda_{k,n}-k/n)/(k-k/n)$.

\subsection{Exact value for 3-UMs}
\label{subsec:3um-incompatibility}

For $k=3$, algebraic and spectral $k$-UMs coincide and \cref{app:primal} shows that \cref{eq:conditional-projector} holds in this case.
\cref{thm:maximal-eigenspace-parent} therefore implies that, for every 3-UM $\mathsf P$ with $n$ outcomes, we have
\begin{equation}
 \eta^{\mg}(\mathsf P)=\frac{\lambda_{3,n}}3,
 \qquad\text{where the explicit formula}\qquad
 \lambda_{3,n}=1+\frac2{\sqrt{n}}\cos\left(\frac13\arccos\frac1{\sqrt{n}}\right)
 \label{eq:triple-value}
\end{equation}
can be obtained with the following argument.
Using \cref{eq:mu3}, the non-zero roots of $\mu_{3,n}$ satisfy, after the shift $t\mapsto1+t$, $t^3-(3/n)t-2/n^2=0$.
The substitution $t=2\cos(\theta)/\sqrt n$ then gives $\cos(3\theta)=1/\sqrt n$, and the largest solution yields \cref{eq:triple-value}.
For the four-outcome construction from \cref{subsec:Hadamard-Clifford}, we have
\begin{equation}
 \lambda_{3,4}=1+\cos\frac{\pi}{9}
 \qquad\text{and}\qquad
 \eta^{\mg}(P,Q,R)=\frac{1+\cos\frac{\pi}{9}}{3}\approx0.64656.
\end{equation}

\section{Universal optimisation}
\label{sec:universal}

The previous section concerns the robustness of a fixed tuple.
The universal problem asks whether any $k$-tuple with the same number of outcomes can be more incompatible.
Whenever a $k$-UM $\mathsf P$ satisfies the assumption of \cref{thm:maximal-eigenspace-parent}, its robustness gives the universal upper bound:
\begin{equation}
 \chi_k^{\mg}(n)\leq\eta^{\mg}(\mathsf P)=\frac{\lambda_{k,n}}k.
 \label{eq:universal-fixed-bracket}
\end{equation}
\cref{subsec:universal-upper} below shows that the selected-sum upper bound in \cref{eq:selected-upper} is minimised by $k$-UMs in several regimes, whereas the sum-of-squares hierarchy in \cref{subsec:universal-sos} supplies universal lower bounds on $\chi_k^{\mg}(n)$.
If these two bounds coincide, one can  conclude that $k$-UMs are maximally incompatible.

\subsection{Optimality of the selected-sum upper bound}
\label{subsec:universal-upper}

Recall that, for arbitrary rank-one projective measurements in $\mathbb C^n$, $\mu_{k,n}$ is the average characteristic polynomial of all selected sums.
This has the following consequence.

\begin{theorem}
 \label{thm:rank-one-spectral}
 Every \emph{rank-one} $k$-tuple of $n$-outcome projective measurements satisfies
 \begin{equation}
 \max_{\vj}\|S_{\vj}\|_\infty\geq\lambda_{k,n}.
 \end{equation}
 A spectral $k$-UM attains equality.
\end{theorem}

A proof is given in \cref{app:spectral}; note that the result is a special case of the mixed-characteristic-polynomial framework of~\cite{MSS15}.
We expect \cref{thm:rank-one-spectral} to remain valid at arbitrary outcome rank.
Partial results in this direction are proved in \cref{app:spectral}: the bound holds for pairs, for any number of dichotomic measurements, and for arbitrary triples.
Beyond these regimes, the generalisation of \cref{thm:rank-one-spectral} to arbitrary rank remains open.

These statements concern the upper bound $\max_{\vj}\|S_{\vj}\|_\infty/k$.
They do not provide a lower bound on $\eta^{\mg}$ for an arbitrary tuple: proving global optimality still requires a universal parent construction.
This is the question studied in the next subsection.

\subsection{Universal sum-of-squares hierarchy}
\label{subsec:universal-sos}

The hierarchy introduced in~\cite{Des26} searches for a parent measurement represented by a positive non-commutative polynomial in one selected outcome from each measurement.
Its positivity and the marginal inequalities are enforced by a sum-of-squares Gram matrix.
The permutation symmetry of settings and outcomes is exploited by block diagonalising the Gram representation~\cite{GP04,IR22}.
A concise formal definition as well as some details on the implementation and numerical results are given in \cref{app:sos}.

Let $\chi_{k,t}^{\mg}(n)$ denote the optimal universal lower bound on the generalised robustness obtained at level $t$ of the hierarchy, that is, using a universal parent measurement that is a polynomial of degree $2t$; \cite{Des26} shows that this value is a lower bound on $\chi_k^{\mg}(n)$.
Combining this lower bound with \cref{eq:universal-fixed-bracket} yields, whenever a $k$-UM exists that satisfies the conditions of \cref{thm:maximal-eigenspace-parent},
\begin{equation}
 \chi_{k,t}^{\mg}(n)\leq\chi_k^{\mg}(n)\leq\frac{\lambda_{k,n}}k.
\end{equation}
A numerical match between the two bounds is evidence of global optimality, see \cref{tab:best-hierarchy-values}.

\begin{table*}[t]
 \centering
 \setlength{\tabcolsep}{5pt}
 \begin{tabular}{|c|ccccccc|}
  \hline
  \diagbox{$~k$}{$n~$} & 2                                                       & 3                                                       & 4                                                       & 5                                                       & 6                                                       & 7                                                            & 8                                                       \\ \hline
                       & \cc                                                     & \cd                                                     & \cd                                                     & \cd                                                     & \cd                                                     & \cd                                                          & \cd                                                     \\
  \multirow{-2}{*}{2}  & \cc\multirow{-2}{*}{\hphantom{$^{[4]}$}$0.85355^{[4]}$} & \cd\multirow{-2}{*}{\hphantom{$^{[4]}$}$0.78868^{[4]}$} & \cd\multirow{-2}{*}{\hphantom{$^{[4]}$}$0.75000^{[4]}$} & \cd\multirow{-2}{*}{\hphantom{$^{[4]}$}$0.72361^{[4]}$} & \cd\multirow{-2}{*}{\hphantom{$^{[4]}$}$0.70412^{[4]}$} & \cd\multirow{-2}{*}{\hphantom{$^{[4]}$}$0.68898^{[4]}$}      & \cd\multirow{-2}{*}{\hphantom{$^{[4]}$}$0.67678^{[4]}$} \\
                       & \cc                                                     & \ce                                                     & \ce                                                     & \cf                                                     & \cf                                                     & $\stackcell{}{\hphantom{^{[3]}}0.56561^{[3]}}$               & $\stackcell{}{\hphantom{^{[3]}}0.54963^{[3]}}$          \\
  \multirow{-2}{*}{3}  & \cc\multirow{-2}{*}{\hphantom{$^{[4]}$}$0.78868^{[4]}$} & \ce\multirow{-2}{*}{\hphantom{$^{[4]}$}$0.69888^{[4]}$} & \ce\multirow{-2}{*}{\hphantom{$^{[4]}$}$0.64656^{[4]}$} & \cf\multirow{-2}{*}{\hphantom{$^{[4]}$}$0.61140^{[4]}$} & \cf\multirow{-2}{*}{\hphantom{$^{[4]}$}$0.58574^{[4]}$} & $\stackcell{0.00035}{}$                                      & $\stackcell{0.00051}{}$                                 \\
                       & \cc                                                     & \ce                                                     & \cf                                                     & $\stackcell{}{\hphantom{^{[3]}}0.54696^{[3]}}$          & $\stackcell{}{\hphantom{^{[3]}}0.51778^{[3]}}$          & $\stackcell{}{\hphantom{^{[3]}}0.49535^{[3]}}$               & $\stackcell{}{\hphantom{^{[3]}}0.47747^{[3]}}$          \\
  \multirow{-2}{*}{4}  & \cc\multirow{-2}{*}{\hphantom{$^{[4]}$}$0.75000^{[4]}$} & \ce\multirow{-2}{*}{\hphantom{$^{[4]}$}$0.64656^{[4]}$} & \cf\multirow{-2}{*}{\hphantom{$^{[4]}$}$0.58719^{[4]}$} & $\stackcell{0.00073}{}$                                 & $\stackcell{0.00130}{}$                                 & $\stackcell{0.00182}{}$                                      & $\stackcell{0.00226}{}$                                 \\
                       & \cc                                                     & \ce                                                     & $\stackcell{}{\hphantom{^{[3]}}0.54701^{[3]}}$          & $\stackcell{}{\hphantom{^{[3]}}0.50397^{[3]}}$          & $\stackcell{}{\hphantom{^{[3]}}0.47277^{[3]}}$          & $\stackcell{}{\hphantom{^{[2]}}0.43429^{[2]}}$               & $\stackcell{}{\hphantom{^{[2]}}0.41495^{[2]}}$          \\
  \multirow{-2}{*}{5}  & \cc\multirow{-2}{*}{\hphantom{$^{[4]}$}$0.72361^{[4]}$} & \ce\multirow{-2}{*}{\hphantom{$^{[4]}$}$0.61140^{[4]}$} & $\stackcell{0.00068}{}$                                 & $\stackcell{0.00167}{}$                                 & $\stackcell{0.00257}{}$                                 & $\stackcell{0.01795}{}$                                      & $\stackcell{0.01897}{}$                                 \\
                       & \cc                                                     & \ce                                                     & $\stackcell{}{\hphantom{^{[3]}}0.51788^{[3]}}$          & $\stackcell{}{\hphantom{^{[3]}}0.47283^{[3]}}$          & $\stackcell{}{\hphantom{^{[2]}}0.42631^{[2]}}$          & $\stackcell{}{\hphantom{^{[2]}}0.40093^{[2]}}$               & $\stackcell{}{\hphantom{^{[2]}}0.38102^{[2]}}$          \\
  \multirow{-2}{*}{6}  & \cc\multirow{-2}{*}{\hphantom{$^{[4]}$}$0.70412^{[4]}$} & \ce\multirow{-2}{*}{\hphantom{$^{[3]}$}$0.58574^{[3]}$} & $\stackcell{0.00120}{}$                                 & $\stackcell{0.00250}{}$                                 & $\stackcell{0.01766}{}$                                 & $\stackcell{0.01921}{}$                                      & $\stackcell{0.02027}{}$                                 \\
                       & \cc                                                     & $\stackcell{}{\hphantom{^{[2]}}0.55882^{[2]}}$          & $\stackcell{}{\hphantom{^{[2]}}0.48472^{[2]}}$          & $\stackcell{}{\hphantom{^{[2]}}0.43629^{[2]}}$          & $\stackcell{}{\hphantom{^{[2]}}0.40187^{[2]}}$          & $\stackcell{}{\hphantom{^{[2]}}0.37600^{[2]}}$               & $\stackcell{}{\hphantom{^{[2]}}0.35574^{[2]}}$          \\
  \multirow{-2}{*}{7}  & \cc\multirow{-2}{*}{\hphantom{$^{[3]}$}$0.68898^{[3]}$} & $\stackcell{0.00714}{}$                                 & $\stackcell{0.01245}{}$                                 & $\stackcell{0.01595}{}$                                 & $\stackcell{0.01827}{}$                                 & $\stackcell{0.01983}{}$                                      & $\stackcell{0.02091}{}$                                 \\
                       & \cc                                                     & $\stackcell{}{\hphantom{^{[2]}}0.54280^{[2]}}$          & $\stackcell{}{\hphantom{^{[2]}}0.46702^{[2]}}$          & $\stackcell{}{\hphantom{^{[2]}}0.41770^{[2]}}$          & $\stackcell{}{\hphantom{^{[2]}}0.38275^{[2]}}$          & $\stackcell{}{\hphantom{^{[2]}}0.35655^{[2]}}$               & $\stackcell{}{\hphantom{^{[2]}}0.33607^{[2]}}$          \\
  \multirow{-2}{*}{8}  & \cc\multirow{-2}{*}{\hphantom{$^{[3]}$}$0.67678^{[3]}$} & $\stackcell{0.00734}{}$                                 & $\stackcell{0.01271}{}$                                 & $\stackcell{0.01622}{}$                                 & $\stackcell{0.01854}{}$                                 & $\stackcell{0.02010}{}$                                      & $\stackcell{0.02116}{}$                                 \\
  \hline
 \end{tabular}
 \caption{
  Best numerical lower bounds $\chi^\mg_{k,t}(n)$ obtained across the computed levels of the hierarchy.
  The superscript indicates the highest level at which $\chi^\mg_{k,t}(n)$ was computed.
  When the lower bound matches $\lambda_{k,n}/k$ at the displayed precision, the grey background records the first matching level: from light grey for level one to dark grey for level four.
  Otherwise, the second line reports the gap between $\lambda_{k,n}/k$ and the numerical lower bound.
 }
 \label{tab:best-hierarchy-values}
\end{table*}

For $(k,n)=(3,4)$, the third level was solved in high precision and returns a numerical estimate of the universal lower bound agreeing with $[1+\cos(\pi/9)]/3$ to approximately $3\times10^{-25}$.
Together with the explicit triple, this numerically brackets $\chi_3^{\mg}(4)$ at that precision.
The computation, however, does not constitute an analytical proof since we could not use it to extract an exact certificate.
For $(k,n)=(3,5),(3,6),(4,4)$, the fourth level was only solved at standard precision, and the observed numerical values lie slightly above the expected $\lambda_{k,n}/k$ (by $5\times10^{-8}$, $1\times10^{-6}$, and $3\times10^{-6}$, respectively).
These small overshoots are consistent with numerical error at the scale of these very large problems, whose sparse SDPA files have sizes approximately $1\mathrm{GB}$, $1\mathrm{GB}$, and $19\mathrm{GB}$, respectively.
The broader pattern motivates the following conjecture.

\begin{conjecture}[Asymptotic sum-of-squares hierarchy value]
 \label{conj:hierarchy-limit}
 For every $k,n\geq2$,
 \begin{equation}
  \lim_{t\to\infty}\chi_{k,t}^{\mg}(n)=\frac{\lambda_{k,n}}k,
  \qquad\text{so that}\qquad
  \chi_3^{\mg}(n)=\frac{\lambda_{3,n}}3
 \end{equation}
 whenever a 3-UM with $n$ outcomes exists, and every 3-UM with those parameters is maximally incompatible.
 In particular, all Hadamard--Clifford triples constructed in \cref{thm:construction} are maximally incompatible.
\end{conjecture}

The values in \cref{tab:best-hierarchy-values} can also be transferred to the dimension-bounded problem: as described in~\cite{Des26}, a universal parent for projective measurements with a fixed number of outcomes yields a parent for rank-one measurements in the corresponding fixed dimension.
The present work could therefore have applications to genuine high-dimensional steering~\cite{DSR+21,Des22}.

\section{Conclusion}

Higher-order unbiasedness behaves differently once one leaves the rank-one setting.
The notion of algebraic $k$-UMs introduced here preserves the hereditary cyclic structure of $k$-fold unbiased bases~\cite{Far17}, while the notion of spectral $k$-UMs isolates the selected-sum data that enter incompatibility bounds.
For triples the notions coincide, and despite the rank-one and three-outcome obstructions, the Hadamard--Clifford construction produces infinitely many higher-rank examples, beginning with a 3-UM with four outcomes of rank two on $\mathbb C^8$.
For this triple as well as for every 3-UM with $n$ outcomes, the maximal-eigenspace construction shows that the generalised incompatibility robustness is equal to $\lambda_{3,n}/3$.

Beyond this fixed-tuple value, 3-UMs are also universally optimal for different incompatibility-related questions.
Among arbitrary projective triples, they minimise the selected-sum upper bound, but this alone does not exclude a different tuple with smaller robustness.
The universal sum-of-squrares hierarchy approaches the problem from the opposite direction by constructing parents valid for all projective tuples.
For $(k,n)=(3,4)$, its high-precision value meets $\lambda_{3,4}/3$ to about $3\times10^{-25}$, giving strong numerical evidence that this triple is globally maximally incompatible.
The broader numerical pattern suggests that the polynomial $\mu_{k,n}$ may govern not only specially structured tuples but the universal asymptotics of the hierarchy.

Two mathematical problems therefore stand out.
On the structural side, no genuinely higher-order example with $k>3$ and $n>2$ is presently known, and the relation between algebraic and spectral $k$-UMs is open beyond the regimes established here.
On the optimisation side, a convergence proof yielding $\lim_t\chi_{k,t}^{\mg}(n)=\lambda_{k,n}/k$ would turn the observed numerical pattern into a universal incompatibility theorem, even for parameter pairs where no finite-dimensional $k$-UM is known to exist.
Outside the low-degree regimes established in \cref{app:primal}, we have not established the required reductions of the moments of $S_{\vj}$ for the maximal-eigenspace parent POVM construction, and establishing maximal incompatibility (or lack thereof) requires further work.

The explicit higher-rank triples also open questions beyond robustness itself.
Their symmetry and rigidity suggest looking for operational characterisations analogous to those known for MUBs and MUMs, including prepare-and-measure self-testing, Bell certification, and steering applications~\cite{FK19,TFR+21,DSR+21,Des22}.
In particular, it would be interesting to know whether correlations can self-test the algebraic 3-UM relations, thereby giving higher-order unbiasedness a directly observable characterisation rather than one defined only through operator identities.

\section*{Acknowledgements}

This work started when S.D.~visited M.F.~in York in December 2025; this trip was financially supported by the QuantERA II Programme that has received funding from the European Union's Horizon 2020 research and innovation programme under Grant Agreement No 101017733 (VERIqTAS).

The notion of algebraic $k$-UMs and the idea of exploring them for maximal incompatibility using the sum-of-squrares method were conceived by the authors.
Through iterative human-directed discussions, GPT 5.6 Sol was used to find the Hadamard--Clifford construction, to derive some of the technical statements, to write numerical code, and to generate the first versions of this paper, which was then reworked by the authors.
The authors take full responsibility for the correctness of the results in this manuscript.

\bibliography{DF26}

\appendix
\crefalias{section}{appendix}
\crefalias{subsection}{appendix}

\section{Relations between algebraic and spectral definitions}
\label{app:relations}

This appendix proves the five statements in \cref{thm:relations}.
Throughout, all measurements have $n$ outcomes and act on the same finite-dimensional Hilbert space $\cH$.
For $J\subseteq[k]$ with $\abs{J}=\ell\geq2$ and $x\in J$, let $\algrel{J}{x}$ denote the assertion that, for every choice of outcomes $(j_y)_{y\in J}$,
\begin{equation}
 \sum_{(x_2,\ldots,x_\ell)}
 A_{j_x|x}A_{j_{x_2}|x_2}\cdots A_{j_{x_\ell}|x_\ell}A_{j_x|x}
 =\frac{(\ell-1)!}{n^{\ell-1}}A_{j_x|x},
 \label{eq:algebraic-subtuple}
\end{equation}
where the sum runs over all orderings of $J\setminus\{x\}$.
The subtuple indexed by $J$ is an algebraic $\ell$-UM precisely when $\algrel{J}{x}$ holds for every $x\in J$.
Thus the original tuple is an algebraic $k$-UM if and only if $\algrel{[k]}{x}$ holds for every $x\in[k]$, while heredity means that $\algrel{J}{x}$ holds for every $J\subseteq[k]$ with $\abs{J}\geq2$ and every $x\in J$.

\textbf{Proof of \cref{thm:relations}, item~\ref{itm:relations-heredity}.}
Let $\mathsf A=(A_{a|x})$ be an algebraic $k$-UM.
Suppose that $\algrel{J}{x}$ holds for a set $J$ with $\abs{J}=\ell\geq3$ and $x\in J$, and choose $y\in J\setminus\{x\}$.
Sum the relation $\algrel{J}{x}$ over $j_y$.
Completeness removes $A_{j_y|y}$, and every ordering of $J\setminus\{x,y\}$ occurs exactly $\ell-1$ times, according to the former position of $y$.
The left-hand side therefore becomes $\ell-1$ times the left-hand side of $\algrel{J\setminus\{y\}}{x}$.
The summed right-hand side is $n(\ell-1)!A_{j_x|x}/n^{\ell-1}$, so division by $\ell-1$ gives the factor $(\ell-2)!/n^{\ell-2}$ required by that relation.
For every $J\subseteq[k]$ and $x\in J$, starting from $\algrel{[k]}{x}$ and eliminating the settings outside $J$ proves $\algrel{J}{x}$.

For two settings, the two choices of the fixed setting give $A_{a|x}A_{b|y}A_{a|x}=A_{a|x}/n$ and $A_{b|y}A_{a|x}A_{b|y}=A_{b|y}/n$.
Summing over $b$ in the first identity and over $a$ in the second shows that all effects are projections.
Taking traces then gives $\Tr A_{a|x}=\Tr A_{b|y}$, and varying settings and outcomes proves common outcome rank.
These pair relations are exactly \cref{eq:mum}, which proves item~\ref{itm:relations-heredity}.

\textbf{Proof of \cref{thm:relations}, item~\ref{itm:relations-rank-one}.}
Consider an algebraic $k$-UM with rank one.
By item~\ref{itm:relations-heredity}, it consists of rank-one PVMs on $\cH\cong\mathbb C^n$ and every sub-$\ell$-tuple is an algebraic $\ell$-UM.
For every setting $x$, choose an orthonormal basis $(\lvert\varphi_a^x\rangle)_{a\in[n]}$ such that $A_{a|x}=\lvert\varphi_a^x\rangle\langle\varphi_a^x\rvert$.
Fix an outcome string $\vj$ and write $\lvert v_x\rangle=\lvert\varphi_{j_x}^x\rangle$.
For $J=\{x_1,\ldots,x_\ell\}\subseteq[k]$ with $\ell\leq n$, the Leibniz formula gives
\begin{equation}
 D_J(\vj_J)
 \coloneqq\det\bigl(\langle v_x|v_y\rangle\bigr)_{x,y\in J}
 =\sum_{\pi\in\fS_\ell}\varepsilon(\pi)\prod_{i=1}^{\ell}\langle v_{x_i}|v_{x_{\pi(i)}}\rangle,
 \label{eq:leibniz}
\end{equation}
where $\vj_J=(j_x)_{x\in J}$, $\fS_\ell$ is the symmetric group on $[\ell]$, and $\varepsilon(\pi)$ is the signature of the permutation $\pi$.
Denoting by $c(\pi)$ the number of (disjoint) cycles of $\pi$, since $m$-cycles have signature $(-1)^{m-1}$, we have $\varepsilon(\pi)=(-1)^{\ell-c(\pi)}$.

Now group the terms in $\sum_{\pi\in\fS_\ell}$ in \cref{eq:leibniz} according to the partition $\mathcal P=\{K_1,K_2,\ldots,K_{\abs{\mathcal{P}}}\}$ of $J$ formed by the cycle supports of the permutations.
For $K\subseteq J$, let $\operatorname{Cyc}(K)$ denote the set of permutations of $K$ consisting of one cycle.
This gives
\begin{equation}
 \begin{aligned}
  D_J(\vj_J)&=\sum_{\mathcal P\vdash J}\sum_{\substack{\pi\in\fS_\ell\\ \pi\sim\mathcal{P}}}(-1)^{\ell-c(\pi)}\prod_{i=1}^{\ell}\langle v_{x_i}|v_{x_{\pi(i)}}\rangle=\sum_{\mathcal P\vdash J}(-1)^{\ell-\abs{\mathcal{P}}}\sum_{\sigma_1\in\operatorname{Cyc}(K_1)}\cdots\sum_{\sigma_{\abs{\mathcal P}}\in\operatorname{Cyc}(K_{\abs{\mathcal P}})}\prod_{j=1}^{\abs{\mathcal P}}\prod_{x\in K_j}\langle v_{x}|v_{\sigma_j(x)}\rangle\\
  &=\sum_{\mathcal P\vdash J}(-1)^{\ell-\abs{\mathcal P}}\prod_{K\in\mathcal P}\left(\sum_{\sigma\in\operatorname{Cyc}(K)}\prod_{x\in K}\langle v_x|v_{\sigma(x)}\rangle\right)
  =\sum_{\mathcal P\vdash J}(-1)^{\ell-\abs{\mathcal P}}\prod_{K\in\mathcal P}\frac{(\abs K-1)!}{n^{\abs K-1}}
  =\sum_{\pi\in\fS_\ell}\left(\frac{-1}{n}\right)^{\ell-c(\pi)},
  \label{eq:rank-one-cycle-grouping}
 \end{aligned}
\end{equation}
where $\mathcal P\vdash J$ denotes partitions of $J$ and $\pi\sim\mathcal{P}$ means that the cycle support structure of $\pi$ is given by $\mathcal P$.
The second equation in \cref{eq:rank-one-cycle-grouping} is a consequence of every permutation being uniquely defined by its disjoint cycles.
The third is a simple factorisation.
For the fourth, notice that substituting the rank-one projectors into $\algrel{K}{x}$ [see \cref{eq:algebraic-subtuple}] gives
\begin{equation}
 \mathcal C_K(\vj_K)
 \coloneqq\sum_{\sigma\in\operatorname{Cyc}(K)}\prod_{x\in K}\langle v_x|v_{\sigma(x)}\rangle
 =\frac{(\abs K-1)!}{n^{\abs K-1}},
\end{equation}
which also holds for a singleton $K$, since its fixed point contributes~1.
For the last equality in \cref{eq:rank-one-cycle-grouping}, observe that $\prod_{K\in\mathcal P}n^{\abs K-1}=n^{\ell-\abs{\mathcal P}}$ and that $\prod_{K\in\mathcal P}(\abs K-1)!$ is the number of permutations whose cycle supports are the blocks of $\mathcal P$, which allows us to convert the sum over $\mathcal P$ back into a sum over $\fS_\ell$.
Introducing the Stirling number of the first kind \cite{Sta11}, $s(\ell,k)\coloneqq(-1)^{\ell-k}\abs{\{\pi\in\fS_\ell:c(\pi)=k\}}$, we obtain
\begin{equation}
 D_J(\vj_J)
 =\frac{1}{n^\ell}\sum_{k=0}^\ell s(\ell,k)n^k
 =\frac{n(n-1)\cdots(n-\ell+1)}{n^\ell}
 \eqqcolon\gamma_{\ell,n},
 \label{eq:rank-one-gram-value}
\end{equation}
where the second equality follows from~\cite[Prop.~1.3.7]{Sta11}.

Let $V_{\vj}=(\lvert v_1\rangle\ \cdots\ \lvert v_k\rangle)\in\bC^{n\times k}$, so that $S_{\vj}=V_{\vj}V_{\vj}^*$, and let $V_{\vj,J}$ denote the submatrix with columns indexed by $J$.
By the Cauchy--Binet formula, we have
\begin{equation}
 \det(t\id_n-V_{\vj}V_{\vj}^*)
 =t^n\det\left[\begin{pmatrix}\id_n&V_{\vj}\end{pmatrix}\cdot\begin{pmatrix}\id_n\\\frac{-1}{t}V_{\vj}^*\end{pmatrix}\right]
 =\sum_{\ell=0}^{\min\{k,n\}}(-1)^\ell t^{n-\ell}
 \sum_{\substack{J\subseteq[k]\\\abs{J}=\ell}}
 \underbrace{\det(V_{\vj,J}^*V_{\vj,J})}_{D_J(\vj_J)}.
 \label{eq:cauchy-binet}
\end{equation}
Substituting \cref{eq:rank-one-gram-value} yields
\begin{equation}
 \det(t\id-S_{\vj})
 =\sum_{\ell=0}^{\min\{k,n\}}(-1)^\ell t^{n-\ell}\binom{k}{\ell}\gamma_{\ell,n}
 =\mu_{k,n}(t),
\end{equation}
which concludes the proof that rank-one algebraic $k$-UMs are spectral.

Conversely, consider a rank-one spectral $k$-UM.
Keep the vector notation $\lvert\varphi^x_{j_x}\rangle$ from above.
For $J\subseteq[k]$ with $\abs{J}=\ell\leq n$ and $x\in J$, we will first show that the completeness of the basis of setting $x$ gives
\begin{equation}
 \sum_{j_x=1}^nD_J(\vj_J)
 =(n-\ell+1)D_{J\setminus\{x\}}(\vj_{J\setminus\{x\}}).
 \label{eq:rank-one-gram-average}
\end{equation}
To see this, let $W$ be the span of the other $\ell-1$ selected vectors and let $\Pi_W$ be the orthogonal projector onto $W$.
If those vectors are dependent, both sides of \cref{eq:rank-one-gram-average} vanish.
Otherwise, by the geometric interpretation of the Gram determinant as a squared volume,
\begin{equation}
 D_J(\vj_J)=D_{J\setminus\{x\}}(\vj_{J\setminus\{x\}})\bigl\|(\id-\Pi_W)\lvert\varphi_{j_x}^x\rangle\bigr\|^2.
\end{equation}
Summing over $j_x$ gives $\sum_{j_x=1}^n\bigl\|(\id-\Pi_W)\lvert\varphi_{j_x}^x\rangle\bigr\|^2=\Tr(\id-\Pi_W)=n-\ell+1$, which proves \cref{eq:rank-one-gram-average}.

Set $\bar{D}_J(\vj_J)=D_J(\vj_J)-\gamma_{\ell,n}$ and define, for a function of the outcomes, its conditional average in $j_x$ to be the uniform average $n^{-1}\sum_{j_x=1}^n$ with all other outcomes fixed.
We prove by induction on $\ell$ that every $\bar{D}_J(\vj_J)$ vanishes, which will later imply the algebraic $k$-UM relations.
For $\ell=1$, $\bar{D}_J(\vj_J)$ vanishes because $D_{\{x\}}(\vj_{\{x\}})=1=\gamma_{1,n}$.
Suppose that the conclusion has been proved for sets of size $\ell-1$.
Since $\gamma_{\ell,n}=(n-\ell+1)\gamma_{\ell-1,n}/n$, \cref{eq:rank-one-gram-average} then shows that $\bar{D}_J(\vj_J)$ has zero conditional average in every variable on which it depends.
On the other hand, comparing the coefficient of $t^{n-\ell}$ in $\mu_{k,n}$ with \cref{eq:cauchy-binet} gives
\begin{equation}\label{eq:coeff-comparison}
 \sum_{\substack{J\subseteq[k]\\\abs{J}=\ell}}{D}_J(\vj_J)=\binom{k}{\ell}\gamma_{\ell,n},
 \qquad\text{so that}\qquad
 \sum_{\substack{J\subseteq[k]\\\abs{J}=\ell}}\bar{D}_J(\vj_J)=0
\end{equation}
for every $\vj$.
Equip functions on $[n]^k$ with the inner product obtained by uniformly averaging their pointwise product over all $\vj\in[n]^k$.
If $J,K\subseteq[k]$ are distinct and $\abs{J}=\abs{K}=\ell$, choose $x\in J\setminus K$.
The function $\bar{D}_K$ is independent of $j_x$, while the conditional average of $\bar{D}_J$ in $j_x$ is zero, so $\bar{D}_J$ and $\bar{D}_K$ are orthogonal.
From \cref{eq:coeff-comparison} it follows that
\begin{equation}
 0=\Big\lVert\sum_{\substack{J\subseteq[k]\\\abs{J}=\ell}}\bar{D}_J(\vj_J)\Big\rVert_2^2
 =\sum_{\substack{J\subseteq[k]\\\abs{J}=\ell}}\left\lVert\bar{D}_J(\vj_J)\right\rVert_2^2,
\end{equation}
where the second equality is a consequence of the orthogonality we have just proved.
Thus $D_J(\vj_J)=\gamma_{\ell,n}$ for every $J$ with $\abs{J}\leq n$ and every choice of outcomes.
This result extends trivially to subsets $J$ such that $\abs{J}>n$, because in this case $D_J(\vj_J)=\gamma_{\ell,n}=0$.

It remains to recover the algebraic $k$-UM relations.
This will follow from the proof of \cref{eq:CJ} below, which we achieve by induction on $\abs{J}=\ell$.
For $\ell=2$, we have $D_J(\vj_J)=1-\abs{\langle v_x|v_y\rangle}^2$.
Hence $D_J(\vj_J)=\gamma_{2,n}=1-1/n$ gives $\mathcal C_J(\vj_J)=\abs{\langle v_x|v_y\rangle}^2=1/n$.
Now let $\ell\geq3$ and assume inductively that $\mathcal C_K(\vj_K)=(\abs K-1)!/n^{\abs K-1}$ for every proper subset $K\subsetneq J$ with $\abs K\geq2$.
Group the terms in the Leibniz expansion of $D_J(\vj_J)$ according to the partition of $J$ formed by their cycle supports, as in \cref{eq:rank-one-cycle-grouping}.
For every partition with at least two blocks, each non-singleton block is a proper subset of $J$ and is therefore covered by the induction hypothesis, while singleton blocks contribute~1.
Consequently, the contribution of every such partition agrees with its value in the calculation leading to \cref{eq:rank-one-gram-value}.
The only contribution not yet determined is that of the one-block partition $\{J\}$, namely $(-1)^{\ell-1}\mathcal C_J(\vj_J)$.
Since $D_J(\vj_J)=\gamma_{\ell,n}$ and \cref{eq:rank-one-gram-value} reaches the same value, the remaining contributions must also agree.
It follows that
\begin{equation}
 \mathcal C_J(\vj_J)=\frac{(\ell-1)!}{n^{\ell-1}}.
 \label{eq:CJ}
\end{equation}

Finally, take $J=[k]$ and fix any $x\in[k]$.
Every cycle with support $[k]$ can be written uniquely as $(x,x_2,\ldots,x_k)$, where $(x_2,\ldots,x_k)$ is an ordering of $[k]\setminus\{x\}$.
After substituting the rank-one projectors into \cref{eq:algebraic-subtuple}, we obtain $\mathcal C_{[k]}(\vj_{[k]})A_{j_x|x}=(k-1)!/n^{k-1}A_{j_x|x}$ by \cref{eq:CJ}, which is nothing else but $\algrel{[k]}{x}$.
This concludes the proof that rank-one spectral $k$-UMs are algebraic.

\textbf{Proof of \cref{thm:relations}, item~\ref{itm:relations-pairs}.}
By definition, algebraic 2-UMs are exactly pairs of MUMs.

Now we show that every pair of MUMs is a spectral 2-UM.
For two selected rank-$r$ projections $P,Q$ from a pair of MUMs, let $V_P,V_Q:\mathbb C^r\to\cH$ be isometries onto their ranges.
Multiplying the MUM relation $PQP=P/n$ by $V_P^*$ on the left and $V_P$ on the right gives $(V_P^*V_Q)(V_P^*V_Q)^*=V_P^*QV_P=\id_r/n$, so all singular values of $V_P^*V_Q$ are $1/\sqrt n$.
Using these singular values, denoting $V=(V_P\ V_Q)$, the equality $V^*V=\bigl(\begin{smallmatrix}\id_r&V_P^*V_Q\\(V_P^*V_Q)^*&\id_r\end{smallmatrix}\bigr)$ gives that the eigenvalues of $V^*V$ are $1\pm1/\sqrt n$.
These eigenvalues, up to zeros, match those of $VV^*=P+Q$, and therefore the non-zero eigenvalues of $P+Q$ are $1\pm1/\sqrt n$, each with multiplicity $r$, and hence
\begin{equation}
 \det(t\id-P-Q)
 =\left(t^{n-2}\left[(t-1)^2-\frac1n\right]\right)^{r}
 =\mu_{2,n}(t)^r.
 \label{eq:mub-pair-char}
\end{equation}

Conversely, take a spectral 2-UM $(P,Q)$.
The argument above works in the reverse direction and shows that all singular values of $V_P^*V_Q$ are $1/\sqrt n$.
This implies $(V_P^*V_Q)(V_P^*V_Q)^*=\id_r/n$, and by multiplying this from the left by $V_P$ and from the right by $V_P^*$ we obtain $PQP=P/n$.
A symmetric argument gives $QPQ=Q/n$, which concludes the proof of item~\ref{itm:relations-pairs}.

\textbf{Proof of \cref{thm:relations}, item~\ref{itm:relations-binary}.}
For a binary PVM, write $P_{a|x}=(\id+aA_x)/2$ for $a\in\{\pm1\}$, where the $A_x$ are Hermitian unitaries, that is, they satisfy $A_x^*=A_x$ and $A_x^2=\id$.
Now take an algebraic $k$-UM.
By item~\ref{itm:relations-heredity} all subpairs are pairs of MUMs.
The MUM relation for two settings is equivalent to $A_xA_y+A_yA_x=0$.
Hence this algebraic $k$-UM is generated by pairwise anticommuting unitaries.
Conversely, given a tuple of $k$ PVMs such that the corresponding $A_x$ anticommute pairwise, for every fixed setting $x$, we have that
\begin{equation}
 \sum_{(x_2,\ldots,x_k)}P_{j_{x_2}|x_2}\cdots P_{j_{x_k}|x_k}
 =\frac{(k-1)!}{2^{k-1}}\left(\id+\sum_{y\neq x}j_yA_y\right),
 \label{eq:binary-symmetrised-product}
\end{equation}
because when the products are summed over all orderings, every term containing at least two distinct unitaries cancels with another term, while the constant and linear terms occur in every ordering.
Since anticommutation also gives $P_{j_x|x}A_yP_{j_x|x}=0$ for $y\neq x$, multiplying \cref{eq:binary-symmetrised-product} by $P_{j_x|x}$ on the left and right yields \cref{eq:algebraic} with $n=2$.
Thus pairwise anticommutation is equivalent to being an algebraic $k$-UM.

Now consider a spectral $k$-UM and set $T_{\vj}=\sum_xj_xA_x$, so $S_{\vj}=(k\id+T_{\vj})/2$.
By definition, because $\mu_{k,2}(t)=t^2-kt+k(k-1)/4$, the equality $\mu_{k,2}(S_{\vj})=0$ by the Cayley--Hamilton theorem implies $T_{\vj}^2=k\id$ for every sign string.
Expanding gives $T_{\vj}^2=k\id+\sum_{x<y}j_xj_y(A_xA_y+A_yA_x)$; multiplying by $j_uj_v$ and averaging over all sign strings isolates the $(u,v)$ coefficient and proves pairwise anticommutation.
Conversely, pairwise anticommutation gives $T_{\vj}^2=k\id$.
Moreover, conjugation of $A_x$ by any $A_y$ with $y\neq x$ changes its sign, so every $A_x$, and hence every $T_{\vj}$, is traceless.
The eigenvalues $\pm\sqrt k$ of $T_{\vj}$ therefore have equal multiplicity, which gives the characteristic polynomial $\mu_{k,2}(t)^r$ for $S_{\vj}$.
This proves item~\ref{itm:relations-binary}.

\textbf{Proof of \cref{thm:relations}, item~\ref{itm:relations-triples}.}
The case $n=2$ is already covered by item~\ref{itm:relations-binary}, so assume $n\geq3$.
Let $(P,Q,R)$ be an algebraic 3-UM and fix selected projections $P=P_a$, $Q=Q_b$, and $R=R_c$.
By item~\ref{itm:relations-heredity}, they have common rank $r$ and are pairwise MUM.
Choose isometries $V_P,V_Q,V_R:\mathbb C^r\to\cH$ onto their ranges.
The MUM relations imply that each of $\sqrt nV_P^*V_Q$, $\sqrt nV_P^*V_R$, and $\sqrt nV_Q^*V_R$ is unitary.
Replacing the isometries by $V_PU_P,V_QU_Q,V_RU_R$ for suitable unitaries $U_P,U_Q,U_R$ does not change their ranges and allows us to arrange $V_P^*V_Q=V_P^*V_R=\id_r/\sqrt{n}$ and $V_Q^*V_R=\Omega/\sqrt{n}$ for a unitary $\Omega$.
Setting $V=(V_P\ V_Q\ V_R)$, we have
\begin{equation}
 V^*V
 =\frac1{\sqrt n}
 \begin{pmatrix}
 \sqrt n\id&\id&\id\\
 \id&\sqrt n\id&\Omega\\
 \id&\Omega^*&\sqrt n\id
 \end{pmatrix}.
 \label{eq:block-gram}
\end{equation}
Multiplying \cref{eq:algebraic-three} by $V_P^*$ on the left and $V_P$ on the right gives $\Omega+\Omega^*=2\id/\sqrt n$.
In an eigenbasis of $\Omega$, the matrix $V^*V$ is the direct sum of $r$ matrices of size $3\times3$, each with characteristic polynomial
\begin{equation}
 t^3-3t^2+\frac{3(n-1)}nt-\frac{(n-1)(n-2)}{n^2}.
 \label{eq:cubic}
\end{equation}
Since $P+Q+R=VV^*$, the operators $V^*V$ and $P+Q+R$ have the same non-zero eigenvalues.
\cref{eq:cubic} is the degree-three factor of $\mu_{3,n}$ [see \cref{eq:mu3}], while the remaining $(n-3)r$ dimensions of $\cH$ contribute zero eigenvalues.
Therefore $\det(t\id-P_a-Q_b-R_c)=\mu_{3,n}(t)^r$.
Thus every algebraic 3-UM is spectral.

Conversely, every selected sum of a spectral 3-UM has largest eigenvalue $\lambda_{3,n}$, so the triple attains equality in \cref{thm:triple-rigidity}.
For this converse we use the independent rigidity result proved in \cref{app:spectral}.
Hence the triple is an algebraic 3-UM, proving item~\ref{itm:relations-triples}.

\section{No-go theorems}
\label{app:nogo}

This appendix proves the non-existence and packing results stated in \cref{sec:kums}.
We treat the rank-one classification first, then the arbitrary-rank obstruction for three outcomes, and finally the packing bound.

\textbf{Proof of \cref{thm:rank-one-no-go}.}
Suppose that three rank-one measurements form an algebraic 3-UM, and write their bases as $(\lvert e_a\rangle)_a$, $(\lvert f_b\rangle)_b$, and $(\lvert g_c\rangle)_c$.
Pairwise unbiasedness allows us to define the phase matrices $u_{ab}=\sqrt{n}\langle e_a|f_b\rangle$, $v_{bc}=\sqrt{n}\langle f_b|g_c\rangle$, and $w_{ca}=\sqrt{n}\langle g_c|e_a\rangle$.
Fix $c$ and set $T^{(c)}_{ab}=u_{ab}v_{bc}w_{ca}$.
The matrix $T^{(c)}$ is complex Hadamard, since it is obtained from $(u_{ab})$ by multiplying rows and columns by phases~\cite{TZ06}.
Taking the trace in \cref{eq:algebraic-three} gives $\operatorname{Re}T^{(c)}_{ab}=1/\sqrt{n}$ for every $a,b$.
This together with the Hadamard property implies that $T^{(c)}=J_n/\sqrt{n}+\mathrm i\sqrt{1-1/n}E^{(c)}$, where $J_n$ is the $n\times n$ all-ones matrix and $E^{(c)}$ has entries in $\{-1,+1\}$.
Completeness of the basis $(\lvert f_b\rangle)_b$ gives $\sum_bT^{(c)}_{ab}=w_{ca}\sum_bu_{ab}v_{bc}=\sqrt{n}$, so every row of $E^{(c)}$ sums to zero, which implies $J_nE^{(c)\mathsf T}=0$.
Using this together with the Hadamard property $T^{(c)}T^{(c)*}=n\id$, we arrive at the formula $(1-1/n)E^{(c)}E^{(c)\mathsf T}=n\id-J_n$.
This formula implies that the inner product of two distinct rows of $E^{(c)}$ is equal to $-n/(n-1)$.
Since the inner product of two $\{\pm1\}$-valued rows is an integer, this is possible only for $n=2$.
Thus no rank-one algebraic 3-UM exists for $n>2$.
By heredity (\cref{thm:relations}, item~\ref{itm:relations-heredity}), no rank-one algebraic $k$-UM exists for $k\geq3$ and $n\geq3$; since rank-one spectral $k$-UMs are algebraic (\cref{thm:relations}, item~\ref{itm:relations-rank-one}), this proves \cref{thm:rank-one-no-go} in this regime.

For $n=2$, the binary equivalence in \cref{app:relations} identifies a rank-one tuple with pairwise anticommuting traceless Hermitian unitaries on $\mathbb C^2$.
They are pairwise orthogonal in the three-dimensional real vector space of traceless Hermitian matrices, so at most three can occur, the three Pauli measurements realising equality.
Finally, mutually unbiased pairs exist in every dimension, for instance from the standard and Fourier bases.

\textbf{Proof of \cref{thm:qutrit-no-go}.}
We present the proof of non-existence of arbitrary-rank algebraic three-outcome $k$-UMs for $k\geq3$.
The argument has two steps.
First, a pair $(P,Q)$ of MUMs is put into a normal form (similarly to Ref.~\cite{FKN23}).
Second, one outcome of a third measurement is parametrised using its MUM relation with $P$, leading to a contradiction.

\paragraph{Normal form.}
Assume that $P=(P_a)_a$ and $Q=(Q_b)_b$ are three-outcome MUMs on $\mathcal H$, with outcomes in $\mathbb Z_3$.
All outcome projections have the same rank, and so there exists a Hilbert space $\cK$ such that $\cH\simeq\cK\otimes\mathbb C^3$.
From Ref.~\cite{FKN23}, we can write
\begin{equation}
 P_a=\id\otimes\lvert a\rangle\langle a\rvert
 \qquad\text{and}\qquad
 Q_b=\frac13\sum_{a,a'\in\mathbb Z_3}V^{(b)}_{aa'}\otimes\lvert a\rangle\langle a'\rvert,
 \label{eq:PQ-normal_0}
\end{equation}
where $V^{(b)}_{aa'}=M_{ab}M_{a'b}^*$ and $M_{ab}$ are the blocks of
\begin{equation}
 M=\begin{pmatrix}\id&\id&\id\\\id&U&-\id-U\\\id&-\id-U&U\end{pmatrix},
\end{equation}
where $U$ is unitary on $\cK$ such that $-\id-U$ is also a unitary.
This in turn implies $U^2+U+\id=0$, $U^*=U^2$, and $U^3=\id$.
Substituting these into \cref{eq:PQ-normal_0}, we obtain the simplified normal form
\begin{equation}
 P_a=\id\otimes\lvert a\rangle\langle a\rvert
 \qquad\text{and}\qquad
 Q_b=\frac13\sum_{a,a'\in\mathbb Z_3}U^{(a-a')b}\otimes\lvert a\rangle\langle a'\rvert.
\end{equation}

\paragraph{Contradiction.}
Suppose now that $(P,Q,R)$ is a three-outcome algebraic 3-UM.
For a fixed outcome $R_c$, mutual unbiasedness with $P$ gives unitaries $V_a$ on $\cK$ such that $R_c=\frac13\sum_{a,a'\in\mathbb Z_3}V_aV_{a'}^*\otimes\lvert a\rangle\langle a'\rvert$ \cite{FKN23}.
For fixed $b$, set $V_a^{(b)}=U^{-ab}V_a$.
Sandwiching the first equality of \cref{eq:algebraic-three} with $\lvert a\rangle$ and multiplying from the left by $U^{-ab}$ and from the right by $U^{ab}$ gives $\sum_{a'\neq a}(V_{a'}^{(b)}V_a^{(b)*}+V_a^{(b)}V_{a'}^{(b)*})=0$.
For a fixed value of $b$ and with the convention $a<a'$, set $H_{aa'}:=V_a^{(b)}V_{a'}^{(b)*}+V_{a'}^{(b)}V_a^{(b)*}$.
For the three values of $a$, the equations therefore read $H_{01}+H_{02}=H_{01}+H_{12}=H_{02}+H_{12}=0$, which implies $H_{aa'}=0$ for all $a<a'$.
Thus $V_a^{(b)}V_{a'}^{(b)*}$ is skew-adjoint whenever $a\neq a'$.

Write $V_{aa'}=V_aV_{a'}^*$.
For $b=0$ we have $V^{(0)}_a=V_a$ and the above skew-adjointness arguments imply that $V_{10}$, $V_{20}$, and $V_{21}=V_{20}V_{10}^*$ are skew-adjoint.
Using the skew-adjointness of $V_{10},V_{20}$ and $V_{21}$ together with the above skew-adjointness arguments for $b=1$, we obtain $U^2V_{10}=V_{10}U$, $UV_{20}=V_{20}U^2$, and $UV_{21}U=U^2V_{21}U^2$.
Taking adjoints in the first of these identities gives $V_{10}^*U=U^2V_{10}^*$, and the second identity then yields $UV_{21}=UV_{20}V_{10}^*=V_{20}U^2V_{10}^*=V_{20}V_{10}^*U=V_{21}U$.
The last identity therefore reduces to $V_{21}U^2=V_{21}U$.
Since $V_{21}$ is unitary, $U^2=U$, so $U=\id$, contradicting $U^2+U+\id=0$.

\textbf{Proof of \cref{prop:packing}.}
For a spectral $k$-UM, since the $\ell=1$ and $\ell=2$ coefficients in \cref{eq:mu} are respectively $-k$ and $k(k-1)(n-1)/(2n)$, Newton's identities imply that the sum of the squares of the zeros of $\mu_{k,n}$ is $k^2-2\cdot k(k-1)(n-1)/(2n)=k+k(k-1)/n$, which gives $\Tr S_{\vj}^2=r[k+k(k-1)/n]$.
Therefore
\begin{equation}
 0=\Tr\Big(\sum_{x}P_{j_x|x}\Big)^2-r[k+k(k-1)/n]=\sum_x\Tr P_{j_x|x}^2+\sum_{\substack{x,y\\ x\neq y}}\Tr(P_{j_x|x}P_{j_y|y})-kr-k(k-1)\frac{r}{n}=\sum_{\substack{x,y\\ x\neq y}}h_{xy}(j_x,j_y),
 \label{eq:trace-unbiasedness}
\end{equation}
where $h_{xy}(a,b)\coloneqq\Tr(P_{a|x}P_{b|y})-r/n$, and we used that $\sum_x\Tr P_{j_x|x}^2=\sum_x\Tr P_{j_x|x}=kr$.
By completeness, summing \cref{eq:trace-unbiasedness} over every outcome $j_z$ for $z\notin\{x,y\}$ isolates $h_{xy}(j_x,j_y)$ because $\sum_a\left[\Tr(P_{a|x}P_{b|y})-r/n\right]=\Tr(P_{b|y})-r=0$, and therefore we obtain $\Tr(P_{a|x}P_{b|y})=r/n$ for $x\neq y$.
For each $x$, the real span $\operatorname{span}_{\mathbb R}\bigl(P_{a|x}-\id/n:a\in[n]\bigr)$ has dimension $n-1$.
Trace-unbiasedness makes these subspaces pairwise Hilbert--Schmidt orthogonal inside the real space of traceless Hermitian operators on $\mathbb C^{nr}$, whose dimension is $(nr)^2-1$.
This proves \cref{eq:packing-main}.

For an algebraic $k$-UM, the pairwise MUM relations give trace-unbiasedness directly.

\section{Hadamard--Clifford calculations}
\label{app:construction}

This appendix verifies the construction in \cref{thm:construction}; standard references for real Hadamard matrices and complex Clifford modules are~\cite{Hor07,LM89}.
Let $H,\widetilde H,E$ and $P,Q,R$ be as in that theorem, with the Hadamard measurements defined by \cref{eq:hadamard-measurement} and the unitary tuple satisfying \cref{eq:clifford-unitaries}.
We first check that the three measurements are projective and pairwise MUM, then verify the algebraic 3-UM relations and the Clifford realisation.
For brevity, write $V_b=V_b(H)$ and $\widetilde V_c=V_c(\widetilde H,E)$, so that $Q_b=V_bV_b^*$ and $R_c=\widetilde V_c\widetilde V_c^*$.

The Hadamard identities imply $V_b^*V_{b'}=\delta_{bb'}\id$ and $\widetilde V_c^*\widetilde V_{c'}=\delta_{cc'}\id$, so $Q$ and $R$ are PVMs.
Their diagonal blocks give $P_aQ_bP_a=P_a/n$ and $P_aR_cP_a=P_a/n$.
Set
\begin{equation}
 C_{bc}=V_b^*\widetilde V_c=\frac1n\sum_{e\in[n]}h_{eb}\widetilde h_{ec}E_e.
\end{equation}
Expanding either $C_{bc}^*C_{bc}$ or $C_{bc}C_{bc}^*$, the diagonal terms give $\id/n$ and the off-diagonal terms cancel in pairs by \cref{eq:clifford-unitaries}.
Hence $C_{bc}^*C_{bc}=C_{bc}C_{bc}^*=\id/n$.
It follows that $Q_bR_cQ_b=Q_b/n$ and $R_cQ_bR_c=R_c/n$.

Using $P_aV_b=h_{ab}(\lvert a\rangle\otimes\id)/\sqrt n$ and $\widetilde V_c^*P_a=\widetilde h_{ac}E_a^*(\langle a\rvert\otimes\id)/\sqrt n$, we find
\begin{align}
 P_a(Q_bR_c+R_cQ_b)P_a
 &=\frac{h_{ab}\widetilde h_{ac}}{n}\lvert a\rangle\langle a\rvert\otimes(C_{bc}E_a^*+E_aC_{bc}^*)
 =\frac{2}{n^2}P_a.
\end{align}
Compressing the other two anticommutators with $V_b^*$ and $\widetilde V_c^*$ gives the same scalar $2/n^2$ by the two identities in \cref{eq:clifford-unitaries}.
Thus the triple satisfies all relations in \cref{eq:algebraic-three} and is an algebraic 3-UM.

Recall from \cref{sec:kums} that a $(k,\ell)$-UM is a $k$-tuple whose every sub-$\ell$-tuple is an $\ell$-UM, with $2\leq\ell\leq k$.
The four-outcome construction above admits a simple extension to a four-outcome $(4,3)$-UM on $\mathbb C^8$.
Let $H_4$ and $E$ be as in \cref{eq:four-outcome-3um}, and define $F=(\id_2,-\mathrm i\sigma_y,\mathrm i\sigma_z,\mathrm i\sigma_x)$.
On $\mathbb C^4\otimes\mathbb C^2$, consider the four rank-two PVMs
\begin{equation}
 P_a=\lvert a\rangle\langle a\rvert\otimes\id_2,
 \qquad
 Q=M(H_4),
 \qquad
 R=M(H_4,E),
 \qquad
 T=M(H_4,F).
\end{equation}
Every sub-3-tuple of $(P,Q,R,T)$ is a 3-UM.
Indeed, under the complex representation $1\mapsto\id_2$, $\mathbf i\mapsto\mathrm i\sigma_x$, $\mathbf j\mapsto-\mathrm i\sigma_y$, and $\mathbf k\mapsto\mathrm i\sigma_z$, the tuples $E$, $F$, and $(E_a^*F_a)_{a\in[4]}$ correspond respectively to $(1,\mathbf i,\mathbf j,\mathbf k)$, $(1,\mathbf j,\mathbf k,\mathbf i)$, and $(1,-\mathbf k,-\mathbf i,-\mathbf j)$.
Each satisfies the same two-sided Clifford relations, so \cref{thm:construction} applies to all sub-3-tuple except $(Q,R,T)$, which can be easily checked to be a 3-UM.
Thus $(P,Q,R,T)$ is a four-outcome rank-two $(4,3)$-UM in dimension eight.
It is, however, neither an algebraic nor a spectral 4-UM: using the 3-UM relations of its subtuples, the algebraic 4-UM relations would require $\Tr(S_{abce}^4)=257/4$ for every four-setting selected sum $S_{abce}=P_a+Q_b+R_c+T_e$.
The spectral target polynomial imposes the same fourth moment, whereas the construction gives the two exact values 64 and 65, depending on the outcome string.

\section{Selected-sum bounds and rigidity}
\label{app:spectral}

This appendix proves the selected-sum statements used in \cref{sec:fixed,sec:universal} and the arbitrary-rank triple rigidity invoked in \cref{app:relations}.
The rank-one and arbitrary-rank arguments are conceptually different.
At rank one, an elementary average-characteristic-polynomial identity identifies $\mu_{k,n}$ directly; after establishing its root structure, the selected-sum bound follows.
At arbitrary rank, we instead introduce a harmonic argument to reduce the problem for $k$ measurements to a trace inequality for $k-1$ measurements.
Pairs and dichotomic tuples admit direct arguments, while the harmonic reduction combined with a low-degree polynomial remainder settles arbitrary triples.
The general arbitrary-rank statement remains open.

\subsection{Rank-one selected-sum bound and polynomial structure}
\label{app:rank-one}

We begin with the average-characteristic-polynomial identity stated in \cref{eq:mu}: for arbitrary rank-one projective measurements, the uniform average of the characteristic polynomials of the selected sums is exactly $\mu_{k,n}$.
We give an elementary proof of this fact, then establish the root properties of $\mu_{k,n}$.
These ingredients make the proof of \cref{thm:rank-one-spectral} immediate.

Let $(\lvert\varphi_a^x\rangle)_{a\in[n]}$ be an orthonormal basis of $\mathbb C^n$ for each setting $x\in[k]$; let $P_{a|x}=\lvert\varphi_a^x\rangle\langle\varphi_a^x\rvert$, and write $S_{\vj}=\sum_xP_{j_x|x}$.
For a fixed outcome string $\vj$, we recall \cref{eq:cauchy-binet}:
\begin{equation}
 \det(t\id-S_{\vj})
 =\sum_{\ell=0}^{\min\{k,n\}}(-1)^{\ell}t^{n-\ell}
 \sum_{\substack{J\subseteq[k]\\\abs{J}=\ell}}
 \det\bigl(\langle\varphi_{j_x}^x|\varphi_{j_y}^y\rangle\bigr)_{x,y\in J}.
 \label{eq:gram-char}
\end{equation}
Thus it remains to compute the uniform average of each Gram determinant.

Fix $J=\{x_1,\ldots,x_\ell\}$ and average successively over the outcomes of these settings.
Suppose first that the vectors selected from $x_1,\ldots,x_{j-1}$ are linearly independent, and let $\mathcal E_{j-1}$ be their span.
Geometrically, a Gram determinant is the squared volume spanned by the corresponding vectors.
When a new vector $\lvert\varphi_{a}^{x_j}\rangle$ is added, this squared volume is multiplied by the squared norm of its component orthogonal to the previous span, namely $\|(\id-\Pi_{\mathcal E_{j-1}})\lvert\varphi_{a}^{x_j}\rangle\|^2$, where $\Pi_{\mathcal E_{j-1}}$ is the orthogonal projector onto $\mathcal E_{j-1}$.
The uniform average of this factor over the $n$ vectors of the $x_j$th basis is
\begin{equation}
 \frac1n\sum_{a=1}^n\|(\id-\Pi_{\mathcal E_{j-1}})\lvert\varphi_{a}^{x_j}\rangle\|^2
 =\frac1n\Tr(\id-\Pi_{\mathcal E_{j-1}})
 =\frac{n-j+1}{n},
 \label{eq:gram-step-average}
\end{equation}
because $\dim\mathcal E_{j-1}=j-1$.
If the previously selected vectors are linearly dependent, their Gram determinant is already zero and remains zero after adjoining another vector.
Iterating \cref{eq:gram-step-average} therefore gives, for a subset $I$ of size $m$,
\begin{equation}
 \frac1{n^m}\sum_{(j_x)_{x\in I}}
 \det\bigl(\langle\varphi_{j_x}^{x}|\varphi_{j_y}^{y}\rangle\bigr)_{x,y\in I}
 =\prod_{j=1}^m\frac{n-j+1}{n}
 =\frac{n!}{(n-m)!n^m}.
 \label{eq:average-gram}
\end{equation}
There are $\binom{k}{\ell}$ subsets $J$ of size $\ell$.
Averaging \cref{eq:gram-char} over all $\vj\in[n]^k$ and using \cref{eq:average-gram} and \cref{eq:mu} therefore yields
\begin{equation}
 \frac1{n^k}\sum_{\vj\in[n]^k}\det(t\id-S_{\vj})
 =\sum_{\ell=0}^{\min\{k,n\}}\frac{(-1)^{\ell}\ell!}{n^{\ell}}\binom{k}{\ell}\binom{n}{\ell}t^{n-\ell}
 =\mu_{k,n}(t).
 \label{eq:average-char}
\end{equation}
This identity is also the rank-one specialisation of the mixed-characteristic-polynomial formula of~\cite[Theorem~4.1]{MSS15}.

We next justify the root properties of $\mu_{k,n}$ quoted after \cref{eq:mu}.
From the definition in \cref{eq:mu},
\begin{equation}
 \mu_{\ell+1,n}=\mu_{\ell,n}-\frac1n\mu'_{\ell,n}.
 \label{eq:mu-recursion}
\end{equation}
We prove inductively that, for $\ell<n$, $\mu_{\ell,n}$ has a zero of multiplicity $n-\ell$ at the origin and $\ell$ positive simple zeros, while for $\ell\geq n$ it has $n$ positive simple zeros.
The claim is immediate for $\ell=1$ since $\mu_{1,n}(t)=t^{n-1}(t-1)$.
Suppose it holds for some $\ell\geq1$, write the positive zeros as $0<\rho_1<\cdots<\rho_s$ with $s=\min\{\ell,n\}$, and set $\nu=\max\{n-\ell,0\}$, so that we have $\mu_{\ell,n}(t)=t^\nu\prod_{i=1}^s(t-\rho_i)$.
For $t$ different from $0,\rho_1,\ldots,\rho_s$, dividing \cref{eq:mu-recursion} by $\mu_{\ell,n}(t)$ then shows that $\mu_{\ell+1,n}(t)=0$ exactly when
\begin{equation}
 n=\frac{\mu'_{\ell,n}(t)}{\mu_{\ell,n}(t)}
 =\frac{\nu}{t}+\sum_{i=1}^s\frac1{t-\rho_i}.
 \label{eq:mu-log-derivative}
\end{equation}
The right-hand side is strictly decreasing on every interval between consecutive poles.
If $\nu>0$, it decreases from $+\infty$ to $-\infty$ on $(0,\rho_1)$ and on each $(\rho_i,\rho_{i+1})$, and from $+\infty$ to 0 on $(\rho_s,\infty)$; it is negative on $(-\infty,0)$.
Thus \cref{eq:mu-log-derivative} has exactly one solution in each of these positive intervals, while \cref{eq:mu-recursion} reduces the multiplicity of the zero at the origin by one, as it shows that $\mu_{\ell+1,n}$ is divisible by $t^{\nu-1}$ but not by $t^\nu$.
If $\nu=0$, there is exactly one solution in each $(\rho_i,\rho_{i+1})$ and one in $(\rho_s,\infty)$, and none below $\rho_1$.
This completes the induction: all zeros are real and non-negative and all positive zeros are simple.


With the polynomial preliminaries in place, \cref{thm:rank-one-spectral} follows immediately.
Suppose, for contradiction, that $\|S_{\vj}\|_\infty<\lambda_{k,n}$ for every outcome string.
Since each $S_{\vj}$ is positive semidefinite, every factor in $\det(\lambda_{k,n}\id-S_{\vj})$ is then strictly positive.
The left-hand side of \cref{eq:average-char} evaluated at $t=\lambda_{k,n}$ would therefore be positive, whereas the right-hand side is $\mu_{k,n}(\lambda_{k,n})=0$.
Hence some selected sum satisfies $\|S_{\vj}\|_\infty\geq\lambda_{k,n}$.
For a rank-one spectral $k$-UM, \cref{eq:spectral} fixes the characteristic polynomial of every selected sum to $\mu_{k,n}$, so every selected sum has norm $\lambda_{k,n}$ and equality is attained.

\subsection{Elementary arbitrary-rank cases}

Two arbitrary-rank regimes admit direct proofs and provide useful benchmarks.
For pairs and, separately, dichotomic $k$-tuples, we will show that
\begin{equation}
 \max_{a,b}\|P_a+Q_b\|_\infty\geq1+\frac1{\sqrt{n}}=\lambda_{2,n}
 \qquad\text{and}\qquad
 \max_{\vj\in\{\pm1\}^k}\Big\|\sum_xP_{j_x|x}\Big\|_\infty
 \geq\frac{k+\sqrt k}{2}=\lambda_{k,2}.
 \label{eq:pair-binary-selected-bounds}
\end{equation}

For the first inequality, let $P=(P_a)_a$ and $Q=(Q_b)_b$ be $n$-outcome PVMs (without zero outcomes for simplicity).
Completeness and the triangle inequality give $1=\|P_a\|_\infty=\|\sum_bP_aQ_bP_a\|_\infty\leq\sum_b\|P_aQ_bP_a\|_\infty$.
Therefore one of the terms $\|P_aQ_bP_a\|_\infty$ must be at least $1/n$.
For two non-zero projections we have $\|P_a+Q_b\|_\infty=1+\sqrt{\|P_aQ_bP_a\|_\infty}$ \cite{DT76}, which then gives the inequality in \cref{eq:pair-binary-selected-bounds}.
If we have equality, this means $\|P_aQ_bP_a\|_\infty\leq1/n$ for all $a,b$, which in turn implies $P_aQ_bP_a\preceq P_a/n$.
But completeness gives $\sum_bP_aQ_bP_a=P_a=\sum_bP_a/n$, and so each positive semidefinite operator $P_a/n-P_aQ_bP_a$ must vanish.
A symmetric argument shows $Q_bP_aQ_b=Q_b/n$, which identifies a pair of MUMs.

For the second inequality, consider $k$ dichotomic PVMs and write $P_{a|x}=(\id+aA_x)/2$ for $a\in\{\pm1\}$ and write $B_{\vj}:=\sum_xj_xA_x$, so $\sum_xP_{j_x|x}=(k\id+B_{\vj})/2$.
Assume, by contradiction, that $\max_{\vj\in\{\pm1\}^k}\|\sum_xP_{j_x|x}\|_\infty<(k+\sqrt k)/2=\lambda_{k,2}$, which implies $(k\id+B_{\vj})/2<(k+\sqrt{k})/2\id$, or equivalently, $B_{\vj}<\sqrt{k}\id$, for all $\vj$.
The same argument for $-\vj$ gives $\|\sum_xP_{-j_x|x}\|_\infty<(k+\sqrt k)/2$, and since $\sum_xP_{-j_x|x}=(k\id-B_{\vj})/2$, the same argument implies $(k\id-B_{\vj})/2<(k+\sqrt{k})\id/2$ and so $-B_{\vj}<\sqrt{k}\id$.
Altogether, we have $\|B_{\vj}\|_{\infty}<\sqrt k$, which implies $B_{\vj}^2<k\id$.
Writing $B_{\vj}^2=\sum_{x,y}j_xj_yA_xA_y$ and averaging over $\vj$, we obtain $2^{-k}\sum_{\vj}B_{\vj}^2=2^{-k}\sum_{\vj}\sum_{x,y}j_xj_yA_xA_y=\sum_{x,y}2^{-k}\sum_{\vj}j_xj_yA_xA_y=\sum_xA_x^2=k\id$ because $2^{-k}\sum_{\vj}j_xj_y=\delta_{xy}$ and because $A_x^2=\id$.
So the average of $B_{\vj}^2$ over $\vj$ is $k\id$, which leads to a contradiction with $B_{\vj}^2<k\id$ for all $\vj$, and thus the inequality in \cref{eq:pair-binary-selected-bounds} must hold.
If equality holds in this inequality, the same argument gives $B_{\vj}^2\preceq k\id$ for every $\vj$.
The same argument also shows that the average of $B_{\vj}^2-k\id$ vanishes, and so we must have $B_{\vj}^2=k\id$ for all $\vj$.
Multiplying the resulting identity by $j_uj_v$ and summing over all sign strings isolates $A_uA_v+A_vA_u$, so equality is equivalent to pairwise anticommutation.

\subsection{Harmonic reduction for arbitrary rank}

The determinant argument in \cref{app:rank-one} above is specific to rank one and does not extend directly to higher-rank outcome projectors.
Instead, we use a different argument, based on the harmonic mean of the residual selected sums defined below.
For an $\ell$-tuple $\mathsf Q=(Q_{a|x})_{a\in[n],x\in[\ell]}$ of $n$-outcome PVMs, write $\widetilde S_{\vj}=\sum_{x=1}^{\ell}Q_{j_x|x}$ for its selected sums.
Whenever $\widetilde S_{\vj}\preceq\lambda\id$ for every $\vj\in[n]^\ell$, define the residual harmonic mean
\begin{equation}
 H_\lambda(\mathsf Q)
 =\left[\frac1{n^\ell}\sum_{\vj\in[n]^\ell}\left(\lambda\id-\widetilde S_{\vj}\right)^{-1}\right]^{-1}.
 \label{eq:residual-harmonic}
\end{equation}
When the measurements are listed explicitly, we use the shorthand $H_\lambda(Q^{(1)},\ldots,Q^{(\ell)})$.
At boundary points, where one of the residuals $\lambda\id-\widetilde S_{\vj}$ may be singular, this expression is defined by the monotone limit obtained by adding a common $\varepsilon\id$ to all residuals and letting $\varepsilon$ decrease to zero.
We use this standard regularisation tacitly and do not display it below~\cite{KA80,Bha97}.

\begin{lemma}[Harmonic reduction]
 \label{thm:harmonic-reduction}
 Fix $k,n\geq2$ and suppose that every $(k-1)$-tuple $\mathsf Q$ of $n$-outcome PVMs on a $d$-dimensional Hilbert space satisfying $\widetilde S_{\vj}\preceq\lambda_{k,n}\id$ obeys
 \begin{equation}
 \Tr H_{\lambda_{k,n}}(\mathsf Q)\leq d.
 \label{eq:general-harmonic-bound}
 \end{equation}
 Then every $k$-tuple of $n$-outcome PVMs satisfies $\max_{\vj}\|S_{\vj}\|_\infty\geq\lambda_{k,n}$.
\end{lemma}

\begin{proof}
 Assume by contradiction that every selected norm of a $k$-tuple is strictly below $\lambda_{k,n}$.
 Fixing an outcome $a$ of the first measurement, we have $S_{\vj}=Q_{a|1}+\widetilde S_{\vj}$, and by the assumption we have $\|Q_{a|1}+\widetilde S_{\vj}\|_\infty<\lambda_{k,n}$ for all $a,\vj$, which also implies $Q_{a|1}+\widetilde S_{\vj}\prec\lambda_{k,n}\id$ for all $a,\vj$.
 In particular, there exists $\delta>0$ such that $\lambda_{k,n}\id-\widetilde S_{\vj}\succeq Q_{a|1}+\delta\id$ for all $a,\vj$.
 The harmonic mean preserves the semidefinite order, as it's a composition of an inverse (which reverses the order), an average (which preserves the order) and another inverse (which reverses the order).
 Thus, we have $H_{\lambda_{k,n}}(\mathsf Q)\succeq Q_{a|1}+\delta\id$.
 Compressing by $Q_{a|1}$, taking traces, and summing over $a$ yields $\Tr H_{\lambda_{k,n}}(\mathsf Q)\geq(1+\delta)d$, contradicting \cref{eq:general-harmonic-bound}.
\end{proof}

\subsection{Arbitrary-rank triples}

The proof has two logically distinct parts.
First, we prove the residual harmonic bound \cref{eq:general-harmonic-bound} for two PVMs.
Second, if a triple attains the bound, equality in the harmonic estimate forces every residual pair to be a pair of MUMs; a $3\times3$ block-Gram argument then forces the algebraic 3-UM relations.

Let $Q=(Q_b)_b$ and $R=(R_c)_c$ be $n$-outcome PVMs on a $d$-dimensional Hilbert space.
Define
\begin{equation}
 \begin{gathered}
 C^{\mathrm{cyc}}_{Q,R}\coloneqq\sum_{b,c}Q_bR_cQ_bR_c,
 \qquad
 \Delta^Q_{bc}\coloneqq Q_bR_cQ_b-\frac1nQ_b,
 \qquad
 \Delta^R_{bc}\coloneqq R_cQ_bR_c-\frac1nR_c,\\
 \Delta^{\mathrm{cyc}}_{Q,R}
 \coloneqq C^{\mathrm{cyc}}_{Q,R}+C^{\mathrm{cyc}*}_{Q,R}-\frac2n\id
 \qquad\text{and}\qquad
 \Delta^{\mathrm{sq}}_{Q,R}
 \coloneqq\sum_{b,c}(\Delta^Q_{bc})^2+\sum_{b,c}(\Delta^R_{bc})^2.
 \end{gathered}
 \label{eq:pair-defects}
\end{equation}
Subtracting $\id/n$ from $C^{\mathrm{cyc}}_{Q,R}$ gives $C^{\mathrm{cyc}}_{Q,R}-\id/n=\sum_{b,c}\Delta^Q_{bc}R_c=\sum_{b,c}Q_b\Delta^R_{bc}$.
To make the operator estimate explicit, put $D=C^{\mathrm{cyc}}_{Q,R}-\id/n$ and $X_c=\sum_b\Delta^Q_{bc}$.
Orthogonality of the $R_c$ gives $DD^*=\sum_cX_cR_cX_c\preceq\sum_cX_c^2=\sum_{b,c}(\Delta^Q_{bc})^2$.
The second expression, $D=\sum_{b,c}Q_b\Delta^R_{bc}$, similarly gives $D^*D\preceq\sum_{b,c}(\Delta^R_{bc})^2$.
The identity valid for every operator $2(DD^*+D^*D)-(D+D^*)^2=(D-D^*)(D-D^*)^*\succeq0$ shows that $(D+D^*)^2\preceq2(DD^*+D^*D)$.
This together with direct expansion and cyclicity of the trace give
\begin{equation}
 (\Delta^{\mathrm{cyc}}_{Q,R})^2\preceq2\Delta^{\mathrm{sq}}_{Q,R},
 \qquad\text{and}\qquad
 \Tr\Delta^{\mathrm{cyc}}_{Q,R}=\Tr\Delta^{\mathrm{sq}}_{Q,R}\geq0,
 \label{eq:pair-defect-bounds}
\end{equation}
where the last inequality follows from $\Delta^{\mathrm{sq}}_{Q,R}\succeq0$.
Moreover, this inequality is an equality if and only if all $\Delta^Q_{bc}$ and $\Delta^R_{bc}$ vanish, namely if and only if $Q$ and $R$ are MUMs.
We also require a spectral bound on $\Delta^{\mathrm{cyc}}_{Q,R}$.
To obtain it, set $E_c=\sum_bQ_bR_cQ_b$.
Then $0\preceq E_c\preceq\id$, $\sum_cE_c=\id$, and $C^{\mathrm{cyc}}_{Q,R}=\sum_cE_cR_c$.
Thus $C^{\mathrm{cyc}}_{Q,R}C^{\mathrm{cyc}*}_{Q,R}=\sum_cE_cR_cE_c\preceq\sum_cE_c^2\preceq\id$.
Therefore, $-2\id\preceq C^{\mathrm{cyc}}_{Q,R}+C^{\mathrm{cyc}*}_{Q,R}\preceq2\id$ and thus
\begin{equation}
 -2\left(1+\frac1n\right)\id
 \preceq\Delta^{\mathrm{cyc}}_{Q,R}\preceq
 2\left(1-\frac1n\right)\id.
 \label{eq:pair-defect-lower}
\end{equation}

\begin{lemma}
 \label{lem:pair-harmonic}
 Suppose that a pair of PVMs $(Q_b)_b$ and $(R_c)_c$ on a $d$-dimensional Hilbert space satisfy $Q_b+R_c\preceq\lambda_{3,n}\id$ for all $b,c$.
 Then $\Tr H_{\lambda_{3,n}}(Q,R)\leq d$, with equality if and only if $Q$ and $R$ are a pair of MUMs.
\end{lemma}

\begin{proof}
 For a pair of MUMs, the possible eigenvalues of $Q_b+R_c$ are $0$ and $1\pm1/\sqrt n$.
 We therefore consider the polynomial with these roots
 \begin{equation}
 w_n(t)=t\left((t-1)^2-\frac1n\right)^2,
 \qquad\text{and define}\qquad
 h_{2,n}(t)=\frac{w_n(\lambda_{3,n})-w_n(t)}{(\lambda_{3,n}-t)w_n(\lambda_{3,n})}.
 \end{equation}
 Note that $w_n$ is non-negative for $t\geq0$.
 The numerator of $h_{2,n}$ is a polynomial divisible by $\lambda_{3,n}-t$ because it vanishes at $t=\lambda_{3,n}$, so $h_{2,n}$ is a polynomial of degree four since $w_n(\lambda_{3,n})-w_n(t)$ has degree five.
 Moreover,
 \begin{equation}
 \frac1{\lambda_{3,n}-t}-h_{2,n}(t)
 =\frac{w_n(t)}{(\lambda_{3,n}-t)w_n(\lambda_{3,n})}
 \geq0
 \qquad\text{for}\qquad0\leq t<\lambda_{3,n}.
 \label{eq:pair-polynomial-minor}
 \end{equation}
 Thus $h_{2,n}$ gives a polynomial lower bound on the resolvent $1/(\lambda_{3,n}-t)$ that is exact on the MUM spectrum.

 To average this polynomial, projectivity and completeness give
 \begin{equation}
 \begin{aligned}
 \frac1{n^2}\sum_{b,c}(Q_b+R_c)
 =\frac2n\id,
 \qquad
 &\frac1{n^2}\sum_{b,c}(Q_b+R_c)^2
 =\left(\frac2n+\frac2{n^2}\right)\id,\\
 \frac1{n^2}\sum_{b,c}(Q_b+R_c)^3
 =\left(\frac2n+\frac6{n^2}\right)\id,
 \qquad
 &\frac1{n^2}\sum_{b,c}(Q_b+R_c)^4
 =\left(\frac2n+\frac{12}{n^2}+\frac2{n^3}\right)\id
 +\frac1{n^2}\Delta^{\mathrm{cyc}}_{Q,R}.
 \end{aligned}
 \end{equation}
 Indeed, after repeated factors are removed, completeness evaluates every sum of words of length at most three.
 At degree four, the only remaining terms are the alternating words $Q_bR_cQ_bR_c$ and $R_cQ_bR_cQ_b$, whose sum gives the cycle term above, see \cref{eq:pair-defects}.
 Substituting these moments into $h_{2,n}$ and using the definition of $\mu_{3,n}$, see \cref{eq:mu3}, yields after a straightforward but tedious calculation
 \begin{equation}
 \frac1{n^2}\sum_{b,c}h_{2,n}(Q_b+R_c)
 =\id+\kappa_{2,n}\Delta^{\mathrm{cyc}}_{Q,R},
 \qquad\text{where}\qquad
 \kappa_{2,n}
 \coloneqq\frac1{n^2w_n(\lambda_{3,n})}
 =\frac1{\lambda_{3,n}[n(\lambda_{3,n}-1)^2-1]^2}.
 \label{eq:pair-polynomial-average}
 \end{equation}

 We next check that the operator on the right is positive definite.
 Recall from \cref{subsec:3um-incompatibility} that the highest root $\lambda_{3,n}$ of $\mu_{3,n}$ satisfies $(\lambda_{3,n}-1)^3-(3/n)(\lambda_{3,n}-1)-2/n^2$.
 Together with $\lambda_{3,n}>1$, this yields $n(\lambda_{3,n}-1)^2>3$ as well as $\lambda_{3,n}>1+1/(2n)$.
 Consequently, $\kappa_{2,n}<1/(4\lambda_{3,n})$ and $\kappa_{2,n}(4+2/n)<1$.
 Together with \cref{eq:pair-defect-lower}, this gives
 \begin{equation}
  \id+\kappa_{2,n}\Delta^{\mathrm{cyc}}_{Q,R}
  \succeq
  \left[1-2\kappa_{2,n}\left(1+\frac1n\right)\right]\id
  \succ0.
  \label{eq:kappa-bound}
 \end{equation}
 Using \cref{eq:pair-polynomial-minor} and \cref{eq:pair-polynomial-average}, we can bound the harmonic mean as
 \begin{equation}
  H_{\lambda_{3,n}}(Q,R)
  =\left[\frac{1}{n^2}\sum_{b,c}\big(\lambda_{3,n}-(Q_b+R_c)\big)^{-1}\right]^{-1}
  \preceq\left[\frac{1}{n^2}\sum_{b,c}h_{2,n}(Q_b+R_c)\right]^{-1}
  =\left(\id+\kappa_{2,n}\Delta^{\mathrm{cyc}}_{Q,R}\right)^{-1},
 \end{equation}
 where the first equality is the definition of the harmonic mean, the inequality is functional calculus on \cref{eq:pair-polynomial-minor} and the last equality is \cref{eq:pair-polynomial-average}.

 To control the trace of this inverse, we use the operator identity $(\id+X)^{-1}=\id-X+X^2(\id+X)^{-1}$ to get
 \begin{equation}
  \left(\id+\kappa_{2,n}\Delta^{\mathrm{cyc}}_{Q,R}\right)^{-1}
  =\id-\kappa_{2,n}\Delta^{\mathrm{cyc}}_{Q,R}+\kappa_{2,n}^2
  \left(\Delta^{\mathrm{cyc}}_{Q,R}\right)^2
  \left(\id+\kappa_{2,n}\Delta^{\mathrm{cyc}}_{Q,R}\right)^{-1}.
 \end{equation}
 The inverse in the last term commutes with $\Delta^{\mathrm{cyc}}_{Q,R}$ and is bounded above by $[1-2\kappa_{2,n}(1+1/n)]^{-1}\id$ by \cref{eq:pair-defect-lower}.
 Taking traces and using \cref{eq:pair-defect-bounds}, we obtain
 \begin{equation*}
  \resizebox{\linewidth}{!}{$\Tr H_{\lambda_{3,n}}(Q,R)-d\leq-\kappa_{2,n}\Tr\Delta^{\mathrm{cyc}}_{Q,R}+\frac{\kappa_{2,n}^2}{1-2\kappa_{2,n}(1+1/n)}\Tr\left[\left(\Delta^{\mathrm{cyc}}_{Q,R}\right)^2\right]\leq-\kappa_{2,n}\left(1-\frac{2\kappa_{2,n}}{1-2\kappa_{2,n}(1+1/n)}\right)\Tr\Delta^{\mathrm{sq}}_{Q,R}\leq0,$}
 \end{equation*}
 where the last inequality also uses that the coefficient multiplying $\Tr\Delta^{\mathrm{sq}}_{Q,R}$ is strictly negative, a consequence of the bounds on $\kappa_{2,n}$ described above \cref{eq:kappa-bound}.
 Equality therefore forces $\Tr\Delta^{\mathrm{sq}}_{Q,R}=0$, so all $\Delta^Q_{bc}$ and $\Delta^R_{bc}$ vanish and $Q,R$ are MUMs.

 Conversely, if $Q,R$ are MUMs, then $\Delta^{\mathrm{cyc}}_{Q,R}=0$ and $w_n(Q_b+R_c)=0$ for every $b,c$.
 Hence the resolvent agrees with $h_{2,n}$ on every selected sum, and \cref{eq:pair-polynomial-average} gives $H_{\lambda_{3,n}}(Q,R)=\id$.
\end{proof}

\begin{theorem}[Arbitrary-rank triple rigidity]
 \label{thm:triple-rigidity}
 Let $P=(P_a)_a$, $Q=(Q_b)_b$, and $R=(R_c)_c$ be $n$-outcome projective measurements on a finite-dimensional Hilbert space.
 Then
 \begin{equation}
  \max_{(a,b,c)\in[n]^3}\|P_a+Q_b+R_c\|_\infty\geq\lambda_{3,n}.
 \end{equation}
 Equality holds if and only if the triple is a 3-UM.
\end{theorem}

\begin{proof}
 The lower bound follows from \cref{thm:harmonic-reduction,lem:pair-harmonic}.
 Suppose now that equality holds, so $P_a+Q_b+R_c\preceq\lambda_{3,n}\id$ for every $a,b,c$.
 For fixed $a$, order monotonicity of the harmonic mean gives $H_{\lambda_{3,n}}(Q,R)\succeq P_a$.
 Compressing by $P_a$, taking traces, and summing over $a$ yields $\Tr H_{\lambda_{3,n}}(Q,R)\geq d$.
 By \cref{lem:pair-harmonic}, equality holds and $Q,R$ are MUMs.
 Cycling the fixed measurement shows that the triple is pairwise MUM.

 Fix $a$ and an isometry $V_a$ onto $\operatorname{ran}P_a$.
 The pairwise MUM relations allow us to choose, coherently for all $b,c$, isometries $W_b$ and $Z_c$ onto $\operatorname{ran}Q_b$ and $\operatorname{ran}R_c$ such that $V_a^*W_b=V_a^*Z_c=\id/\sqrt n$.
 Writing $W_b^*Z_c=\Omega_{abc}/\sqrt n$ defines a unitary $\Omega_{abc}$, and the block Gram matrix of $V_a,W_b,Z_c$ has the form of \cref{eq:block-gram}.
 For an eigenvalue $e^{\mathrm i\theta}$ of $\Omega_{abc}$, the corresponding scalar block has characteristic polynomial $t^3-3t^2+3(1-n^{-1})t-[1-3n^{-1}+2n^{-3/2}\cos\theta]$.
 Its largest root is strictly increasing in $\cos\theta$ and equals $\lambda_{3,n}$ when $\cos\theta=1/\sqrt n$.
 The upper bound on the selected-sum spectrum therefore implies $\Omega_{abc}+\Omega_{abc}^*\preceq2\id/\sqrt n$.
 Moreover, $V_a^*(Q_bR_c+R_cQ_b)V_a=(\Omega_{abc}+\Omega_{abc}^*)/n^{3/2}$, so summing over $b,c$ and using completeness gives
 \begin{equation}
  \frac1{n^2}\sum_{b,c}(\Omega_{abc}+\Omega_{abc}^*)=\frac2{\sqrt n}\id.
 \end{equation}
 The positive semidefinite operators $2\id/\sqrt n-\Omega_{abc}-\Omega_{abc}^*$ have zero average and hence vanish individually.
 Multiplying by $V_a$ from the left and $V_a^*$ from the right gives precisely \cref{eq:algebraic-three} with $P_a$ fixed; cycling the fixed measurement gives the other two families of relations, which concludes.
 Conversely, the block-Gram calculation in \cref{app:relations} shows that every algebraic 3-UM is spectral and therefore attains $\lambda_{3,n}$.
\end{proof}

\section{Primal attainment and conditional moments}
\label{app:primal}

This appendix proves \cref{thm:maximal-eigenspace-parent}, following the strategy introduced for MUBs in~\cite{DSFB19}.
The first subsection isolates the maximal-eigenspace criterion itself.
The second explains under which assumptions $k$-UMs verify the identity required in \cref{eq:conditional-projector}.
This makes use of so-called \emph{conditional moments}, that is, averages $n^{1-k}\sum_{\vj}\delta_{j_x,a}S_{\vj}^m$ for some degree $m$.

\subsection{Maximal-eigenspace parent criterion}

Let $\mathsf P=(P_{a|x})_{a\in[n],x\in[k]}$ be a projective tuple, set $S_{\vj}=\sum_xP_{j_x|x}$, and let $\lambda=\max_{\vj\in[n]^k}\|S_{\vj}\|_\infty$.
Let $(G_{\vj})_{\vj\in[n]^k}$ satisfy the constraints in \cref{eq:eta} at visibility $\eta$.
Then
\begin{equation}
 \eta k\dim\cH\leq\sum_{x,a}\Tr\left(P_{a|x}\sum_{\vj}\delta_{j_x,a}G_{\vj}\right)
 =\sum_{\vj}\Tr(S_{\vj}\,G_{\vj})
 \leq\lambda\sum_{\vj}\Tr G_{\vj}=\lambda\dim\cH.
\end{equation}
At equality, every $G_{\vj}$ is supported on the maximal eigenspace of $S_{\vj}$ and every marginal slack is orthogonal to the corresponding outcome projector.

Now let $\mathsf P$ satisfy the hypotheses of \cref{thm:maximal-eigenspace-parent}.
Summing that identity over $a$ gives $n^{1-k}\sum_{\vj}\Pi_{\vj}=\id$, so the operators in \cref{eq:maximal-eigenspace-parent} form a POVM.
For each $a,x$, the same identity gives a marginal larger than $(\lambda_{k,n}/k)P_{a|x}$.
The parent is therefore feasible at visibility $\lambda_{k,n}/k$, while \cref{eq:selected-upper} gives the converse inequality.
This proves \cref{thm:maximal-eigenspace-parent}.

\subsection{Conditional-moment reductions}

For a spectral $k$-UM of common outcome rank $r$, let $p_{k,n}$ be the Lagrange polynomial equal to one at $\lambda_{k,n}$ and zero at the other distinct zeros of $\mu_{k,n}$.
Then $\Pi_{\vj}=p_{k,n}(S_{\vj})$ and $\deg p_{k,n}=\min\{k,n-1\}$; the equivalent power-limit representation of the same maximal-eigenvalue projector was used in~\cite{DSFB19}.
Assume that, for each ${0\leq m\leq\min\{k,n-1\}}$,
\begin{equation}
 \frac1{n^{k-1}}\sum_{\vj}\delta_{j_x,a}S_{\vj}^m=\alpha_mP_{a|x}+\beta_m\frac{\id-P_{a|x}}{n-1},
\end{equation}
with coefficients independent of $a,x$.
Applying $p_{k,n}$ gives the same affine form for the conditional average of $\Pi_{\vj}$.
If its coefficients are $\alpha,\beta$, summing over $a$ gives $(\alpha+\beta)\id$.
The left-hand side has trace $nr=\dim\cH$, because every $\Pi_{\vj}$ has rank $r$, and hence $\alpha+\beta=1$.
Furthermore, $n^k\lambda_{k,n}r=\sum_{\vj}\Tr(S_{\vj}\,\Pi_{\vj})=kn^k\alpha r$, so $\alpha=\lambda_{k,n}/k$ and $\beta=1-\lambda_{k,n}/k$.
Thus the conditional-moment hypothesis implies \cref{eq:conditional-projector}.
In particular, the identity holds in either of the following low-degree regimes:
\begin{enumerate}
 \item $\min\{k,n-1\}\leq4$ and the tuple is pairwise MUM;
 \item $\min\{k,n-1\}\leq5$ and every sub-3-tuple is a 3-UM.
\end{enumerate}
We prove these statements by reducing the relevant words.
Fix $A=P_{a|x}$ and expand a conditional moment into words in the selected projections.
If a setting different from $x$ occurs exactly once, summing its outcome removes that letter by completeness.
Adjacent repetitions are removed by idempotence, while any subword $XYX$ with $X$ and $Y$ belonging to distinct settings reduces to $X/n$ by the MUM relation.

For a word of length at most four, either it involves at most two settings, in which case these reductions apply directly, or it involves at least three settings, in which case some setting different from $x$ occurs exactly once.
Iterating the preceding reductions leaves a scalar multiple of $A$ or $\id$:
\begin{gather}
  \frac1{n^2}\sum_{\vj}\delta_{j_x,a}S_{\vj}=A+\frac2n\id,
  \qquad
  \frac1{n^2}\sum_{\vj}\delta_{j_x,a}S_{\vj}^2=\left(1+\frac4n\right)A+\left(\frac2n+\frac2{n^2}\right)\id,\\
  \frac1{n^2}\sum_{\vj}\delta_{j_x,a}S_{\vj}^3=\left(1+\frac{10}n+\frac6{n^2}\right)A+\left(\frac2n+\frac8{n^2}\right)\id,
  \qquad
  \frac1{n^2}\sum_{\vj}\delta_{j_x,a}S_{\vj}^4=\left(1+\frac{18}n+\frac{32}{n^2}\right)A+\left(\frac2n+\frac{18}{n^2}+\frac{10}{n^3}\right)\id.\nonumber
\end{gather}
This proves the first low-degree regime.

At length five, a word not covered by the same argument must use exactly three settings with multiplicities $1,2,2$, with the fixed setting $x$ occurring once.
Write $B=P_{b|y}$ and $C=P_{c|z}$ for the selected projections from the other two settings.
After removing adjacent repetitions and subwords of the form $XYX$ with $X$ and $Y$ belonging to distinct settings, the only two surviving orientations are $BCABC$ and $CBACB$.
Applying \cref{eq:algebraic-three} to the triple $A,B,C$, with $B$ and $C$ fixed, gives
\begin{equation*}
 BCABC+CBACB=\frac{2}{n^2}(BC+CB)-\frac1n(BAC+CAB)
 \qquad\text{and}\qquad
 \frac1{n^2}\sum_{b,c}(BCABC+CBACB)=\frac4{n^4}\id-\frac2{n^3}A.
\end{equation*}
The first identity follows by multiplying the identities $B(CA+AC)B=2B/n^2$ and $C(BA+AB)C=2C/n^2$ by the remaining projection on the side needed to form the two length-five words, and then using $CBC=C/n$ and $BCB=B/n$.
The second identity follows by summing over $b,c$ and finally yields the fifth conditional moment:
\begin{equation}
 \frac1{n^2}\sum_{\vj}\delta_{j_x,a}S_{\vj}^5
 =\left(1+\frac{28}n+\frac{98}{n^2}+\frac{26}{n^3}\right)A
 +\left(\frac2n+\frac{32}{n^2}+\frac{52}{n^3}+\frac4{n^4}\right)\id.
\end{equation}
This proves the second low-degree regime.

At length six, repeated cycles such as $ABCABC$ survive these reductions, which is the reason why no general statement is claimed.

\section{Sum-of-squares hierarchy and numerical data}
\label{app:sos}

This appendix describes the universal hierarchy introduced in~\cite{Des26} and used in \cref{sec:universal}, restricting the presentation to triples of measurements.
The construction for arbitrary $k$ is identical, with words over $k$ settings in place of the three-letter words below.
Schematically, at level $t$ one searches for non-commutative sum-of-squares polynomials satisfying
\begin{equation}
 G_{\vj}\in\Sigma_t,
 \qquad\sum_{\vj}G_{\vj}=\id,
 \qquad\sum_{\vj}\delta_{j_x,a}G_{\vj}-\eta P_{a|x}\in\Sigma_t\quad\text{for every }a,x,
\end{equation}
where $\Sigma_t$ denotes the cone represented by Gram matrices on reduced words of length at most $t$.
The remainder of the appendix constructs these Gram representations and reduces them exactly by algebraic and symmetry arguments.
A Julia implementation, reproduction scripts, and instructions for solving the accompanying sparse-SDPA instances are available with this work~\cite{GitHub}.

\subsection{The hierarchy for triples}

Let $P=(P_a)_{a\in[n]}$, $Q=(Q_b)_{b\in[n]}$, and $R=(R_c)_{c\in[n]}$ be arbitrary projective measurements.
We eliminate their last outcomes through $P_n=\id-\sum_{a<n}P_a$, $Q_n=\id-\sum_{b<n}Q_b$, and $R_n=\id-\sum_{c<n}R_c$.
Products are then reduced using projectivity and orthogonality: two adjacent projectors from the same measurement collapse when their outcomes agree and vanish otherwise.

For three abstract projections $A,B,C$, let $\mathcal W_t$ be the reduced words of length at most $t$, where only adjacent repetitions are removed, and let $\mathbf w_t(A,B,C)$ be the column vector formed by these words.
There are $3\cdot2^t-2$ such words.
For example,
\begin{equation*}
 \resizebox{\linewidth}{!}{$\mathbf w_3(A,B,C)=(\id,A,B,C,AB,AC,BA,BC,CA,CB,ABA,ABC,ACA,ACB,BAB,BAC,BCA,BCB,CAB,CAC,CBA,CBC)^{\mathsf T}.$}
\end{equation*}
A positive semidefinite Gram matrix $X$ of size $(3\cdot2^t-2)\times(3\cdot2^t-2)$ defines
\begin{equation}
 g_X(A,B,C)=\mathbf w_t(A,B,C)^*X\mathbf w_t(A,B,C),
 \label{eq:sos-parent-Gram}
\end{equation}
a positive non-commutative polynomial of degree at most $2t$~\cite{Hel02,HM04,BKP16}.
Evaluating it on one outcome from each measurement gives candidate parent effects $G_{abc}=g_X(P_a,Q_b,R_c)$.

The normalisation condition $\sum_{a,b,c}G_{abc}=\id$ is imposed as a polynomial identity after reducing all words in the projective-measurement algebra.
For the marginals, one requires each residual $\sum_{b,c}G_{abc}-\eta P_a$, $\sum_{a,c}G_{abc}-\eta Q_b$, and $\sum_{a,b}G_{abc}-\eta R_c$ to admit a Gram representation by reduced concrete words of length at most $t$.
These sum-of-squares identities imply the corresponding operator inequalities in every finite-dimensional representation of the projector relations.
They therefore produce a parent measurement for every projective triple, independently of the Hilbert-space dimension.
This is the Gram, or dual, counterpart of the non-commutative moment relaxations used for quantum correlations and polynomial optimisation~\cite{NPA08,PNA10,TPBA24,BKP16,Des26}.

We denote the level-$t$ optimum by $\chi_{3,t}^{\mg}(n)$.
Padding the Gram matrices with zero rows and columns embeds level $t$ into level $t+1$, and hence
\begin{equation}
 \chi_{3,t}^{\mg}(n)\leq\chi_{3,t+1}^{\mg}(n)\leq\chi_3^{\mg}(n).
\end{equation}
The first inequality is exact, whereas the numerical values reported below are subject to floating-point errors in the computed symmetry decomposition and in the SDP solve.

\subsection{An exact reduction of the marginal word basis}
\label{app:sos-sieve}

A direct marginal certificate uses all reduced concrete words of length at most $t$ after eliminating one outcome of each measurement.
For $(n,t)=(4,3)$, there are nine explicit projector letters and, after the first letter, six choices belonging to a different measurement.
The resulting Gram vector therefore contains $1+9+9\cdot6+9\cdot6^2=388$ words.

This basis can be reduced without weakening the hierarchy.
Consider the residual for the first outcome of the first measurement, $\Delta=\sum_{b,c}G_{1bc}-\eta P_1$.
It belongs to the subalgebra generated by $P_1$ and all explicit outcomes of $Q$ and $R$.
Define a unital $*$-homomorphism $\rho$ on the projective-measurement algebra by
\begin{equation}
 \rho(P_1)=P_1,\qquad\rho(P_a)=0\quad(2\leq a<n),\qquad\rho(Q_b)=Q_b,\qquad\rho(R_c)=R_c.
\end{equation}
Since the eliminated projector is $P_n=\id-\sum_{a<n}P_a$, its image is $\rho(P_n)=\id-P_1$, so the projective-measurement relations are preserved.
Moreover, $\rho(\Delta)=\Delta$.
If $\Delta=\mathbf z^*Y\mathbf z$ with $Y\succeq0$ on the full word vector, then $\Delta=\rho(\mathbf z)^*Y\rho(\mathbf z)$, which is again a sum of squares and involves only the smaller subalgebra.
Thus it is sufficient to retain one projector from the selected measurement and every explicit projector from the other two measurements.

For four outcomes this leaves one $P$-letter, three $Q$-letters, and three $R$-letters.
Let $u_\ell$ count length-$\ell$ words ending in $P_1$ and let $s_\ell$ count those ending in a $Q$- or $R$-letter.
Then $u_1=1$, $s_1=6$, and $u_{\ell+1}=s_\ell$, $s_{\ell+1}=6u_\ell+3s_\ell$.
The numbers of words of lengths one, two, and three are consequently 7, $30$, and $132$, so the marginal Gram vector has
\begin{equation}
 1+7+30+132=170
\end{equation}
entries instead of $388$.
This reduction is an exact reduction, not a numerical sparsification.

\subsection{Symmetrisation and block diagonalisation}
\label{app:sos-symmetry}

Relabelling settings and outcomes preserves the objective and all polynomial constraints, so every feasible solution can be averaged without changing $\eta$.
For the parent polynomial, this symmetrisation replaces $g_X$ from \cref{eq:sos-parent-Gram} by
\begin{equation}
 g_X^{\mathrm{sym}}(A,B,C)
 =\frac16\sum_{\pi\in\fS_3}
 \mathbf w_t(A_{\pi(1)},A_{\pi(2)},A_{\pi(3)})^*X
 \mathbf w_t(A_{\pi(1)},A_{\pi(2)},A_{\pi(3)}),
\end{equation}
where $(A_1,A_2,A_3)=(A,B,C)$.
After one outcome has been eliminated, the normalisation equations retain the action of $(\fS_{n-1})^3\rtimes\fS_3$.
For the selected marginal, the relevant stabiliser is $\fS_{n-2}\times((\fS_{n-1})^2\rtimes\fS_2)$.
The coefficient identities therefore need to be imposed only on one representative of each word orbit.

The same group actions preserve the two Gram spaces.
Their commutant algebras are block diagonalised by the usual Wedderburn reduction for invariant semidefinite programmes~\cite{GP04,DPS07,Val09,BGSV12,IR22}.
Concretely, after an orthogonal change of basis an invariant Gram matrix has the form
\begin{equation}
 U^*XU=\bigoplus_\lambda\bigl(X_\lambda\otimes\id_{d_\lambda}\bigr)
 \qquad\text{and}\qquad
 X\succeq0\ \Longleftrightarrow\ X_\lambda\succeq0\ \text{for every }\lambda,
 \label{eq:sos-Wedderburn}
\end{equation}
so only the smaller multiplicity matrices $X_\lambda$ enter the SDP.
The implementation constructs the permutation representations on the word bases, diagonalises a generic self-adjoint element of each commutant, and aligns equivalent irreducible components through intertwining maps.
The latter identify the bases of equivalent irreducible subspaces, so that each copy carries the same representation matrices and the tensor-product form in \cref{eq:sos-Wedderburn} is explicit.
It checks orthogonality, the transformed generator actions, invariance of reconstructed matrices, and the dimension of the self-adjoint part of the commutant before accepting a decomposition.

For $(k,n,t)=(3,4,3)$, the $22\times22$ parent Gram matrix is reduced to positive semidefinite blocks of sizes 7, 5, and 3.
After the exact sieve above, the $170\times170$ marginal Gram matrix is reduced to blocks of sizes $20$, $14$, $12$, 7, 7, 2, and 1.
These reductions are essential for the high-precision computation.

\subsection{Implementation and numerical results}

The hierarchy is implemented in Julia and modelled through JuMP~\cite{BEKS17,LDG+23}.
The reduced SDP is exported in sparse SDPA format and read back before being solved with Mosek or Hypatia for extended precision~\cite{CKV22}.
The implementation records the Gram-block eigenvalues, polynomial residuals, symmetry-decomposition validation errors, and primal-dual consistency checks together with the objective value~\cite{GitHub}.

\cref{tab:sos-sizes-n4} illustrates the effect of the two reductions for four-outcome triples.
For the marginal Gram matrix, the first dimension uses all reduced concrete words after eliminating one outcome of each measurement, whereas the second uses the sieve from \cref{app:sos-sieve}.
The equation columns compare the full coefficient supports with the non-zero orbit representatives retained in the SDP.
The reduction becomes increasingly pronounced with the level.
For fixed $t$, the parent Gram space is independent of $n$, and the computed data show a rapid stabilisation of the reduced marginal blocks and coefficient equations as $n$ grows: at level two the largest block and both equation counts are already constant from $n=3$ onwards, and at level three they are constant from $n=4$ onwards.

\begin{table*}[t]
 \centering
 \setlength{\tabcolsep}{2.5pt}
 \begin{tabular}{c c c c c c c c c c c c c c c}
  \toprule
  & \multicolumn{8}{c}{Gram matrix} & \multicolumn{6}{c}{Equations} \\
  \cmidrule(lr){2-9}\cmidrule(lr){10-15}
  & \multicolumn{3}{c}{Parent} & \multicolumn{5}{c}{Marginal} & \multicolumn{3}{c}{Normalisation} & \multicolumn{3}{c}{Marginal} \\
  \cmidrule(lr){2-4}\cmidrule(lr){5-9}\cmidrule(lr){10-12}\cmidrule(lr){13-15}
  \multirow{2}{*}{~$t$~}
  & \multirow{2}{*}{$\abs{\mathcal W_t}$} & & largest & full & & sieved & & largest & \multirow{2}{*}{before} & & \multirow{2}{*}{after} & \multirow{2}{*}{before} & & \multirow{2}{*}{after} \\
  & & & block & dim. & & dim. & & block & & & & & & \\
  \midrule
  1 & 4  & $\rightarrow$ & 2  & 10   & $\rightarrow$ & 8   & $\rightarrow$ & 3  & 1           & $\rightarrow$ & 1     & 38         & $\rightarrow$ & 6    \\
  2 & 10 & $\rightarrow$ & 3  & 64   & $\rightarrow$ & 38  & $\rightarrow$ & 6  & 703         & $\rightarrow$ & 7     & 746        & $\rightarrow$ & 37   \\
  3 & 22 & $\rightarrow$ & 7  & 388  & $\rightarrow$ & 170 & $\rightarrow$ & 20 & 37\,693     & $\rightarrow$ & 102   & 14\,282    & $\rightarrow$ & 378  \\
  4 & 46 & $\rightarrow$ & 15 & 2332 & $\rightarrow$ & 746 & $\rightarrow$ & 74 & 2\,141\,695 & $\rightarrow$ & 2949~ & ~273\,050  & $\rightarrow$ & 5225 \\
  \bottomrule
 \end{tabular}
 \caption{
  Sizes of the level-$t$ SDP for triples of four-outcome measurements before and after the sieve from \cref{app:sos-sieve} and symmetry reduction sketched in \cref{app:sos-symmetry}.
  The post-symmetry equation counts are the non-zero representative equations appearing in the SDP.
 }
 \label{tab:sos-sizes-n4}
\end{table*}

We also provide in \cref{tab:hierarchy-numerical-tables} the detailed values giving rise to \cref{tab:best-hierarchy-values}, making the finite-level realisation of \cref{conj:hierarchy-limit} clear for levels one to three.
At level four, the values are given in \cref{tab:best-hierarchy-values} and all instances that we were able to solve reproduce the conjecture value $\lambda_{k,n}/k$.
The main obstacle for solving more of them is not the symmetrisation but the behaviour of the resulting SDP, which becomes increasingly ill-conditioned as $k$ and $n$ increase.

\begin{table*}[h]
 \centering
 \scalebox{0.75}{%
 \begin{minipage}{1.3\textwidth}
 \centering
 \subfloat[\large Level one: $\chi^\mg_{k,1}(n)$]{%
  \begin{tabular}{|c|ccccccc|}
  \hline
  \diagbox{$~k~$}{$~n~$} & 2            & 3         & 4         & 5         & 6         & 7         & 8         \\ \hline
  2                      & \cc~0.85355~ & ~0.76026~ & ~0.70572~ & ~0.67016~ & ~0.64514~ & ~0.62657~ & ~0.61224~ \\
  3                      & \cc0.78868   & 0.66667   & 0.59686   & 0.55166   & 0.51994   & 0.49640   & 0.47822   \\
  4                      & \cc0.75000   & 0.61436   & 0.53785   & 0.48860   & 0.45412   & 0.42857   & 0.40884   \\
  5                      & \cc0.72361   & 0.58010   & 0.50000   & 0.44868   & 0.41284   & 0.38631   & 0.36583   \\
  6                      & \cc0.70412   & 0.55556   & 0.47329   & 0.42078   & 0.38419   & 0.35714   & 0.33628   \\
  7                      & \cc0.68898   & 0.53690   & 0.45323   & 0.40000   & 0.36298   & 0.33564   & 0.31457   \\
  8                      & \cc0.67678   & 0.52213   & 0.43750   & 0.38381   & 0.34653   & 0.31903   & 0.29785   \\
  \hline
  \end{tabular}%
 }
 \hfill
 \subfloat[\large Level two: $\chi^\mg_{k,2}(n)$]{%
  \begin{tabular}{|c|ccccccc|}
  \hline
  \diagbox{$~k~$}{$~n~$} & 2            & 3            & 4            & 5            & 6            & 7            & 8            \\ \hline
  2                      & \cc~0.85355~ & \cc~0.78868~ & \cc~0.75000~ & \cc~0.72361~ & \cc~0.70412~ & \cc~0.68898~ & \cc~0.67678~ \\
  3                      & \cc0.78868   & 0.69552      & 0.64037      & 0.60323      & 0.57620      & 0.55548      & 0.53899      \\
  4                      & \cc0.75000   & 0.64136      & 0.57782      & 0.53548      & 0.50494      & 0.48171      & 0.46334      \\
  5                      & \cc0.72361   & 0.60518      & 0.53665      & 0.49135      & 0.45887      & 0.43429      & 0.41495      \\
  6                      & \cc0.70412   & 0.57893      & 0.50713      & 0.45996      & 0.42631      & 0.40093      & 0.38102      \\
  7                      & \cc0.68898   & 0.55882      & 0.48472      & 0.43629      & 0.40187      & 0.37600      & 0.35574      \\
  8                      & \cc0.67678   & 0.54280      & 0.46702      & 0.41770      & 0.38275      & 0.35655      & 0.33607      \\
  \hline
  \end{tabular}%
 }
 \par\medskip
 \subfloat[\large Level three: $\chi^\mg_{k,3}(n)$]{%
  \begin{tabular}{|c|ccccccc|}
  \hline
  \diagbox{$~k~$}{$~n~$} & 2            & 3            & 4            & 5            & 6            & 7            & 8            \\ \hline
  2                      & \cc~0.85355~ & \cc~0.78868~ & \cc~0.75000~ & \cc~0.72361~ & \cc~0.70412~ & \cc~0.68898~ & \cc~0.67678~ \\
  3                      & \cc0.78868   & \cc0.69888   & \cc0.64656   & 0.61135      & 0.58555      & 0.56561      & 0.54963      \\
  4                      & \cc0.75000   & \cc0.64656   & 0.58696      & 0.54696      & 0.51778      & 0.49535      & 0.47747      \\
  5                      & \cc0.72361   & \cc0.61140   & 0.54701      & 0.50397      & 0.47277      & {--}         & {--}         \\
  6                      & \cc0.70412   & \cc0.58574   & 0.51788      & 0.47283      & {--}         & {--}         & {--}         \\
  7                      & \cc0.68898   & {--}         & {--}         & {--}         & {--}         & {--}         & {--}         \\
  8                      & \cc0.67678   & {--}         & {--}         & {--}         & {--}         & {--}         & {--}         \\
  \hline
  \end{tabular}%
 }
 \hfill
 \subfloat[\large Conjectured value: $\lambda_{k,n}/k$]{%
  \begin{tabular}{|c|ccccccc|}
  \hline
  \diagbox{$~k~$}{$~n~$} & 2         & 3         & 4         & 5         & 6         & 7         & 8         \\ \hline
  2                      & ~0.85355~ & ~0.78868~ & ~0.75000~ & ~0.72361~ & ~0.70412~ & ~0.68898~ & ~0.67678~ \\
  3                      & 0.78868   & 0.69888   & 0.64656   & 0.61140   & 0.58574   & 0.56596   & 0.55014   \\
  4                      & 0.75000   & 0.64656   & 0.58719   & 0.54769   & 0.51908   & 0.49718   & 0.47973   \\
  5                      & 0.72361   & 0.61140   & 0.54769   & 0.50563   & 0.47534   & 0.45224   & 0.43392   \\
  6                      & 0.70412   & 0.58574   & 0.51908   & 0.47534   & 0.44397   & 0.42014   & 0.40129   \\
  7                      & 0.68898   & 0.56596   & 0.49718   & 0.45224   & 0.42014   & 0.39583   & 0.37665   \\
  8                      & 0.67678   & 0.55014   & 0.47973   & 0.43392   & 0.40129   & 0.37665   & 0.35724   \\
  \hline
  \end{tabular}%
 }
 \end{minipage}%
 }
 \caption{
  Numerical values of $\chi^\mg_{k,t}(n)$, together with the candidate values $\lambda_{k,n}/k$.
  Every hierarchy entry is a numerical estimate of the corresponding universal lower bound.
  A dash indicates that the corresponding instance was not computed, and a grey background that the numerical value matches the conjectured value at the displayed precision at that level.
 }
 \label{tab:hierarchy-numerical-tables}
\end{table*}

\end{document}